\documentclass[letterpaper]{article} 
\usepackage{aaai23}  
\nocopyright
\usepackage{times}  
\usepackage{helvet}  
\usepackage{courier}  
\usepackage[hyphens]{url}  
\usepackage{graphicx} 
\usepackage{natbib}  
\usepackage{caption} 
\usepackage{algorithm}
\usepackage{algorithmic}

\usepackage{newfloat}
\usepackage{listings}
\DeclareCaptionStyle{ruled}{labelfont=normalfont,labelsep=colon,strut=off} 
\floatstyle{ruled}
\newfloat{listing}{tb}{lst}{}
\floatname{listing}{Listing}

\usepackage{amssymb}
\usepackage{amsmath}
\usepackage{amsthm}
\usepackage[hidelinks]{hyperref} 
\usepackage{nicefrac}
\usepackage{tikz}
\usetikzlibrary{arrows,shapes}
\usepackage{colortbl}
\usepackage{booktabs}
\usepackage{cleveref}
\usepackage{array}
\usepackage{graphicx}
\usepackage{xcolor}
\usepackage{enumitem}
\usepackage{tabularx}
\usepackage{mathtools}
\usepackage{xstring}
\usepackage{forloop}
\usepackage[ruled,vlined,linesnumbered,algo2e]{algorithm2e}
\usepackage{thm-restate}
\usepackage{microtype}

\usepackage{longtable}
\usepackage{makecell}
\usepackage{marvosym}

\theoremstyle{definition}
\newtheorem{definition}{Definition}

\theoremstyle{plain}

\newtheorem{lemma}{Lemma}
\newtheorem{fact}{Fact}
\newtheorem{proposition}{Proposition}

\newtheoremstyle{bfnote}
{}{}%
{}{}%
{\bfseries}{.}%
{ }%
{\thmname{#1}\thmnumber{ #2}\thmnote{ (#3)}}
\theoremstyle{bfnote}
\newtheorem{remark}{Remark}

\usepackage[colorinlistoftodos,textsize=tiny]{todonotes}

\newcommand{\patrick}[1]{\todo[color=teal!60,inline]{\textbf{Patrick:} #1}}
\newcommand{\jonas}[1]{\todo[color=cyan!40]{\textbf{Jonas:} #1}}
\newcommand{\jonasinline}[1]{\todo[color=cyan!40,inline]{\textbf{Jonas:} #1}}

\presetkeys{todonotes}{disable}{} 

\usepackage[british]{babel}

\title{Strategy\-proofness and Proportionality in Party-Approval Multiwinner Elections}
\author{
    Th\'eo Delemazure,\textsuperscript{\rm 1}
    Tom Demeulemeester, \textsuperscript{\rm 2}
    Manuel Eberl, \textsuperscript{\rm 3}
    Jonas Israel, \textsuperscript{\rm 4}
    Patrick Lederer \textsuperscript{\rm 5}
}
\affiliations{
    \textsuperscript{\rm 1} Paris Dauphine University, PSL, CNRS, France,\\
    \textsuperscript{\rm 2} Research Center for Operations Research \& Statistics, KU Leuven, Belgium,\\
    \textsuperscript{\rm 3} Computational Logic Group, University of Innsbruck, Austria, \\
    \textsuperscript{\rm 4} Research Group Efficient Algorithms, Technische Universität Berlin, Germany,\\ 
    \textsuperscript{\rm 5} Technical University of Munich, Germany, \\
    \href{mailto:theo.delemazure@dauphine.eu}{theo.delemazure@dauphine.eu}, \href{mailto:tom.demeulemeester@kuleuven.be}{tom.demeulemeester@kuleuven.be}, \href{mailto:manuel.eberl@uibk.ac.at}{manuel.eberl@uibk.ac.at}, \href{mailto:j.israel@tu-berlin.de}{j.israel@tu-berlin.de}, \href{mailto:ledererp@in.tum.de}{ledererp@in.tum.de} 
}

\usepackage{bibentry}

\begin{document}

\maketitle

\begin{abstract}
In party-approval multiwinner elections the goal is to allocate the seats of a fixed-size committee to parties based on the approval ballots of the voters over the parties. In particular, each voter can approve multiple parties and each party can be assigned multiple seats. Two central requirements in this setting are proportional representation and strategy\-proofness. Intuitively, proportional representation requires that every sufficiently large group of voters with similar preferences is represented in the committee. Strategy\-proofness demands that no voter can benefit by misreporting her true preferences. We show that these two axioms are incompatible for anonymous party-approval multiwinner voting rules, thus proving a far-reaching impossibility theorem. 
The proof of this result is obtained by formulating the problem in propositional logic and then letting a SAT solver show that the formula is unsatisfiable. Additionally, we demonstrate how to circumvent this impossibility by considering a weakening of strategy\-proofness which requires that only voters who do not approve any elected party cannot manipulate. While most common voting rules fail even this weak notion of strategy\-proofness, we characterize Chamberlin--Courant approval voting within the class of Thiele rules based on this strategy\-proofness notion.
\end{abstract}

\section{Introduction}

A central problem in multi-agent systems is collective decision making: given the preferences of multiple agents over a set of alternatives, a joint decision has to be made. 
While classic literature for this problem focuses on the case of choosing a single alternative as the winner, there is also a wide range of scenarios where a whole set of winners needs to be elected. For instance, this is the case in parliamentary elections, where the seats of a parliament are assigned to parties based on the voters' preferences. In the literature, parliamentary elections are studied under the term \emph{apportionment} and a crucial assumption for their analysis is that voters are only allowed to vote for a single party \citep{BaYo01a,Puke14a}. However, this assumption has recently been criticized because of its lack of flexibility and expressiveness \citep{BLS18a,BGP+19a}. Following \citet{BGP+19a}, we thus study \emph{party-approval elections}. In this setting, the parliament, or more generally a multiset of fixed size, is elected based on the approval ballots of the voters, i.e., each voter reports a set of approved parties instead of only her most preferred one. 

Two central desiderata for party-approval elections are \emph{proportional representation} and \emph{strategy\-proofness}. The former requires that the chosen committee should proportionally reflect the voters' preferences. The latter postulates that no voter can benefit by misreporting her preferences. While \citet{BGP+19a} have shown that even core-stable committees, which satisfy one of the highest degrees of proportionality, always exist in party-approval elections, strategy\-proof\-ness is not yet well-understood for this setting. We thus analyze the trade-off between strategy\-proof\-ness and proportional representation for party-approval elections in this paper.

Our research question also draws motivation from related models (see \Cref{fig:settings} for details). Firstly, 
party-approval elections can be seen as a special case of \emph{approval-based committee (ABC) elections}, where voters approve individual candidates instead of parties and the outcome is a subset of the candidates instead of a multiset. For ABC elections, proportionality and strategy\-proofness have received significant attention (see, e.g., the survey by \citet{LaSk22a}). Unfortunately, these axioms are jointly incompatible for ABC voting rules \citep{Pete18a} and our study can thus be seen as an attempt to circumvent this impossibility. Even more, there are hints that these axioms could be compatible for party-approval elections: this setting lies logically between ABC elections on the one side, and either \emph{apportionment} (where voters can only approve a single party instead of multiple ones \citep{BaYo01a,Puke14a}) or \emph{fair mixing} (where the outcome is a probability distribution over the parties instead of a multiset \citep{BMS05a,ABM17a}) on the other side. Since strategyproofness and proportionality are compatible in the latter two models, it seems reasonable to conjecture positive results for party-approval elections.

\paragraph{Our contribution.} Unfortunately, it turns out that strategy\-proofness conflicts even with minimal notions of proportional representation in party-approval elections. To prove this, we introduce the notions of weak representation and weak proportional representation, which require that a party is assigned at least $1$ (resp. $\ell$) out of $k$ available seats if it is uniquely approved by at least an $\frac{1}{k}$ (resp. $\frac{\ell}{k}$) fraction of the voters. Then, we show in \Cref{sec:sp+prop} the following impossibility theorems ($k$, $m$, and $n$ denote the numbers of seats, parties, and voters, respectively):
\begin{itemize}
     \item No anonymous party-approval rule satisfies both strategy\-proofness and weak representation if $k\geq 3$, $m\geq k+1$, and $2k$ divides $n$.
     \item No anonymous party-approval rule satisfies both strategy\-proofness and weak proportional representation if $k\geq 3$, $m\geq 4$, and $2k$ divides $n$.
\end{itemize}

\noindent
The first result shows that the incompatibility of strategy\-proofness and proportional representation first observed for ABC elections also prevails for party-approval elections. Even more, our result implies such an impossibility for ABC elections as our setting is more general. The main drawback of the first result is that it requires more parties than seats in the committee. While this assumption is true for many applications inspired from ABC voting, this is not the case in our initial example of parliamentary elections. However, our second impossibility shows that strategy\-proofness is still in conflict with proportional representation if $k>m$. 

We prove both of these results with a computer-aided approach based on SAT solving, which has recently led to a number of sweeping impossibility results \citep[e.g.,][]{BBPS21a,BSS19b}. In particular, our computer proof relies on 635 profiles, which makes it the largest computer proof in social choice theory (the previous record is due to \citet{BBPS21a} and uses 386 profiles).
\jonas{Do we want to state this so prominently here?}
\patrick{I think that this is an impressive sounding fact and thus might help to get the paper accepted. But I have no strong fealings if you want to remove it.}

Finally, to derive more positive results, we investigate in \Cref{sec:wSP} a weakening of strategy\-proofness that requires that voters who do not approve any party in the elected committee cannot manipulate. Perhaps surprisingly, many commonly studied committee voting rules fail this condition. On the other hand, we characterize Chamberlin--Courant approval voting as the only Thiele rule satisfying this strategy\-proofness notion and weak representation, thus proving an attractive escape route to our impossibility results.

\paragraph{Related work.}
Party-approval elections have been introduced by \citet{BGP+19a} who showed that strong proportionality axioms can be satisfied in this setting, but we are not aware of any follow-up paper on this topic yet. We thus draw much inspiration from ABC elections for which there is a large amount of work on proportional representation \citep[e.g.,][]{ABC+16a,SFF+17a,PeSk20a,BIMP22a} and strategy\-proofness \citep[e.g.,][]{AGG+15a,Pete18a,KVV+20a,LaSk18a,botan2021manipulability}. For instance, \citet{ABC+16a} analyze ABC voting rules with respect to more restrictive variants of weak representation. The main message from work on proportional representation is that there are few ABC voting rules that guarantee strong representation axioms. The results on strategy\-proofness are mostly negative: after early results \citep{AGG+15a,LaSk18a} proving that no known rule satisfies both strategy\-proofness and proportional representation, \citet{Pete18a} showed that these axioms are inherently incompatible for ABC voting rules (see also \citep{Dudd14a,KVV+20a} for related results). Our first impossibility is closely connected to this result but logically independent: while we need stronger strategy\-proofness and representation axioms and additionally anonymity, our setting is more flexible and we use no efficiency condition. 

\begin{figure}[t]
\hspace{-0.3cm}
    \scalebox{1}{
      \begin{tikzpicture}
      [block/.style={minimum width = 2.4cm, minimum height = 0.5cm, rectangle, draw, rounded corners, font=\small},
      connec/.style={thick},
      notexist/.style={fill=gray!30},
      line/.style={->,>=latex'}]
        
        \node[block] (a) {ABC};
        \node[above = 0.45cm of a.west, anchor=west, xshift=-0.15cm] (l1) {Models ordered by domain restrictions:};
        \node[block, right = 0.5cm of a] (b) {party-approval};
        \node[block, right = 0.5cm of b] (c) {apportionment};

        \node[block, below = 0.55cm of a] (d) {fair mixing};
        \node[above = 0.45cm of d.west, anchor= west, xshift=-0.15cm] (l1) {Models ordered by output type:};
        \node[block, right = 0.5cm of d] (e) {party-approval};
        \node[block, right = 0.5cm of e] (f) {ABC};
        
        \draw[line, connec] (a.east)--(b.west);
        \draw[line, connec] (b.east)--(c.west);
        \draw[line, connec] (d.east)--(e.west);
        \draw[line, connec] (e.east)--(f.west);
      \end{tikzpicture}
    }
     \caption{Relation of party-approval elections to other voting settings. An arrow from $X$ to $Y$ means that model $X$ is more general than model $Y$. In the settings in the top row, elections return sets of alternatives but the models impose different restrictions on the input profiles: for ABC voting every input profile is allowed, for party-approval profiles each voter can for each party (viewed as a set of alternatives) either approve all of its members or none, and for apportionment each voter must approve all members of exactly one party (see \citet{BGP+19a} for more details). In the bottom row, the models are ordered with respect to their output type: all of fair mixing, party-approval elections, and ABC elections can take arbitrary approval profiles as input, but fair mixing rules return a probability distribution over the alternatives, party-approval rules choose a multiset of the alternatives, and ABC rules choose a subset of the alternatives. In particular, this shows that party-approval elections can be seen both as generalization and special case of ABC elections.}
    \label{fig:settings}
 \end{figure}

\section{Preliminaries}\label{sec:preliminaries}

Let $N=\{1,\dots, n\}$ denote a set of $n$ voters and $\mathcal{P} = \{a,b,c,\dots\}$ a set of $m$ parties. Each voter $i\in N$ is assumed to have a \emph{dichotomous preference relation} over the parties, i.e., she partitions the parties into approved and disapproved ones. The \emph{approval ballot $A_i\subseteq \mathcal P$} of a voter $i$ is the non-empty set of her approved parties. With slight abuse of notation, we omit commas and brackets when writing approval ballots. Let $\mathcal{A}$ denote the set of all possible approval ballots. An \emph{approval profile $A \in \mathcal{A}^n$} is the collection of the approval ballots of all voters. Given an approval profile $A$, the goal in \emph{party-approval elections} is to assign a fixed number of seats to the parties. We call such an outcome a \emph{committee}, which is formally a multiset of parties $W : \mathcal{P} \rightarrow \mathbb{N}$, and $W(x)$ denotes the number of seats assigned to party $x$. We extend this notation also to sets of parties $X\subseteq \mathcal P$ by defining $W(X)=\sum_{x\in X} W(x)$. Furthermore, we indicate specific committees by square brackets, e.g., $[a,a,b]$ is the committee containing party $a$ twice and party $b$ once. Let $\mathcal{W}_k$ denote the set of all committees of size $k$.

A \emph{party-approval rule} is a function $f$ which takes an approval profile $A\in\mathcal{A}^n$ and a target committee size $k$ as input and returns a winning committee $W\in\mathcal{W}_k$. In particular, party-approval rules are resolute, i.e., there is always a single winning committee. We define $f(A,k,x)$ as the number of seats assigned to party $x$ by $f$ for the profile $A$ when choosing a committee of size $k$. Just as for committees, we extend this notion also to sets by defining $f(A,k,X)=\sum_{x\in X} f(A,k,x)$. 

Two well-known properties of voting rules are anonymity and Pareto optimality. Intuitively, anonymity requires that all voters are treated equally, i.e., a party-approval rule $f$ is \emph{anonymous} if $f(A,k)=f(A',k)$ for all committee sizes $k\in \mathbb{N}$ and all approval profiles $A, A' \in \mathcal{A}^n$ such that there is a permutation $\pi:N\rightarrow N$ with $A'_i=A_{\pi(i)}$. 

Next, we say that a party $x$ \emph{Pareto dominates} another party $y$ in an approval profile $A$ if $y\in A_i$ implies $x\in A_i$ for all $i\in N$ and there is a voter $i\in N$ with $x\in A_i$ and $y\not\in A_i$. Then, a party-approval rule $f$ is \emph{Pareto optimal} if $f(A,k,y)=0$ for all approval profiles $A$, committee sizes $k$, and parties $y$ that are Pareto dominated in $A$.

\subsection{Proportional Representation}
One of the central desiderata in committee elections is to choose a committee that proportionally represents the voters' preferences. The notion of \emph{justified representation}, introduced by \citet{ABC+16a}, formalizes this idea by requiring that in a committee of size $k$, any group of voters $G\subseteq N$ with $|G| \geq \frac{n}{k}$ that agrees on a party should be represented. While this is already a rather weak representation axiom, in this paper we will consider a further weakening which we call \emph{weak representation}. Intuitively, weak representation weakens justified representation by only considering cases where all voters in $G$ uniquely approve a single party $x$.

\begin{definition}[Weak Representation]
A party-approval rule $f$ satisfies \emph{weak representation} if $f(A,k,x)\geq 1$ for every profile $A$, committee size $k$, party $x$, and group of voters $G$ such that $|G|\geq \frac{n}{k}$ and $A_i=\{x\}$ for all $i\in G$. 
\end{definition}

Weak representation can be easily satisfied if we have more seats in the committee than parties by simply assigning at least one seat to every party. This, however, clearly contradicts the idea of proportional representation because a large part of the chosen committee is independent of the voters' preferences. To address this issue, we consider weak proportional representation, which is a weakening of \emph{proportional justified representation} \citep{SFF+17a}. Clearly, weak proportional representation implies weak representation.

\begin{definition}[Weak Proportional Representation]
A party-approval rule $f$ satisfies \emph{weak proportional representation} if $f(A,k,x)\geq \ell$ for every profile $A$, committee size $k$, party $x$, and group of voters $G$ such that $|G|\geq \ell \frac{n}{k}$ and $A_i=\{x\}$ for all $i\in G$. 
\end{definition}

\subsection{Strategy\-proofness}

Intuitively, strategy\-proofness requires that a voter cannot benefit by lying about her true preferences. Consequently, if a party-approval rule fails strategy\-proofness, we cannot expect the voters to submit their true preferences, which may lead to socially undesirable outcomes.

\begin{definition}[Strategy\-proofness]
A party-approval rule $f$ is \emph{strategy\-proof} if $f(A,k,A_i)\geq f(A',k,A_i)$ for all approval profiles $A, A'$, committee sizes $k$, and voters $i\in N$ such that $A_j=A'_j$ for all $j\in N\setminus \{i\}$. %We call $f$ \emph{manipulable} if it is not \emph{strategy\-proof}.
\end{definition}

The motivation for this strategy\-proofness notion stems from the assumption that voters are indifferent between their approved parties. Then, only the number of seats assigned to these parties matters to the voters. Also, note that this strategy\-proofness notion is commonly used in ABC voting \citep[e.g.,][]{LaSk18a,botan2021manipulability}, often under the name cardinality strategy\-proofness, and equivalent notions are used for fair mixing \citep[e.g.,][]{BMS05a,ABM17a}.

Since we will show that strategy\-proofness is in conflict with minimal representation axioms, we also consider the following weakening which requires that only voters without representation in the committee cannot manipulate.

\begin{definition}[Strategy\-proofness for Unrepresented Voters] \label{def:weaksp}
A party-approval rule $f$ is \emph{strategy\-proof for unrepresented voters} if $f(A,k,A_i)\geq f(A',k,A_i)$ for all approval profiles $A, A'$, committee sizes $k$, and voters $i\in N$ such that $A_j=A'_j$ for all $j\in N\setminus \{i\}$ and $f(A,k,A_i)=0$.
\end{definition} 

We believe this to be a sensible relaxation of strategy\-proofness because voters without any representation in the committee are more prone to manipulate. Firstly, voters who do have some representation may be more cautious to manipulate because they fear losing their representation when misstating their preferences. Secondly, the benefit of having additional representation in the committee is less straightforward than that of being represented at all.

\subsection{Party-Approval Rules}\label{subsec:rules}

Finally, we introduce three classes of party-approval rules. Note that even though we define party-approval rules for a fixed numbers of voters $n$ and parties $m$, all subsequent rules are independent of such details. 

\paragraph{Thiele rules.}
Thiele rules are arguably the most well-studied class of rules in the ABC setting. Introduced by \citet{Thie95a}, a \emph{$w$-Thiele rule} $f$ is defined by a non-increasing and non-negative vector $w = (w_1, w_2, \dots)$ and chooses for each committee size $k$ the committee $W\in\mathcal{W}_k$ that maximizes the score $s_w(W,A)=\sum_{i\in N} \sum_{j=1}^{W(A_i)} w_j$. Throughout the paper, we assume without loss of generality that $w_1=1$. There are many well-known Thiele rules, such as:
\begin{itemize}[topsep=3pt,itemsep=0pt]
    \item \emph{approval voting (AV)}:  $w = (1,1,1,\ldots)$,
    \item \emph{proportional approval voting (PAV)}: $w = (1,\frac{1}{2},\frac{1}{3},\ldots)$,
    \item \emph{Chamberlin--Courant approval voting (CCAV)}:\\${w=(1,0,0,\ldots)}$.
\end{itemize}

\paragraph{Sequential Thiele rules.}
Sequential Thiele rules are closely related to Thiele rules: instead of optimizing the score of the committee, these rules proceed in rounds and greedily choose in each iteration the party that increases the score of the committee the most. An important example of sequential Thiele rules is \emph{sequential proportional approval voting (\mbox{seqPAV})} (defined by $w=(1, \frac{1}{2}, \frac{1}{3},\dots)$) .

\paragraph{Divisor methods based on majoritarian portioning.}
\citet{BGP+19a} introduced the concept of composite party-approval rules, which combine a \emph{portioning method} with an \emph{apportionment method}. 
In this paper, we focus on an important subclass of such composite rules, namely \emph{divisor methods based on majoritarian portioning}, because many of these rules satisfy strong representation axioms \citep{BGP+19a}.
These methods first apply \emph{majoritarian portioning} to compute a weight $w_x$ for each party $x$. Majoritarian portioning works in rounds and in each round, we determine the party $x$ that is approved by the most voters. Then, we set its weight $w_x$ to the number of voters who approve $x$ and remove all corresponding voters from the profile. This process is repeated until no voters are left. Finally, for all parties $x$ that have no weight after all voters were removed, we set $w_x=0$. After the portioning, we use a \emph{divisor method} to allocate the seats to the parties based on the weights $w_x$. Divisor methods are defined by a monotone function $g:\mathbb{N}_{0}\rightarrow \mathbb{R}_{>0}$ and proceed in rounds: in the $i$-th round, the next seat is assigned to the party $x$ that maximizes $\frac{w_x}{g(t_{i-1}^x)}$, where $t_{i-1}^x$ is the number of seats allocated to $x$ in the previous $i-1$ rounds. An example of a divisor method is Jefferson's method (where $g(x)=x+1$).\medskip

Note that all party-approval rules defined above are in principle irresolute, i.e., they may declare multiple committees as tied winners of an election. Since we investigate resolute voting rules in this paper, we assume that ties are always broken lexicographically. Formally, this means that, for every $k\in \mathbb{N}$, there is a linear tie-breaking order $\succ_k$ on the committees $W\in\mathcal{W}_k$ and, if a party-approval rule $f$ declares multiple committees as tied winners, we choose the best one according to $\succ_k$. Similarly, if any rule is tied between multiple parties in a step, the tie is broken according to $\succ_1$. The assumption of lexicographic tie-breaking is standard in the literature on strategy\-proofness \citep[e.g.,][]{FHH10a,AGG+15a}.

\section{Impossibility Results}\label{sec:sp+prop}

In this section, we discuss the incompatibility of strategy\-proofness and proportional representation for party-approval rules by proving two sweeping impossibility theorems.

\begin{restatable}{theorem}{thmSAT}\label{thm:SAT}
No party-approval rule simultaneously satisfies anonymity, weak representation, and strategy\-proofness if $k\geq 3$, $m\geq k+1$, and $2k$ divides $n$.
\end{restatable}

Note that \Cref{thm:SAT} does not hold for all combinations of $k$, $m$, and $n$: we require that $2k$ divides $n$ and that $m\geq k+1$. The first assumption is mainly a technical one as we---just like other authors \citep{Pete18a,KVV+20a}---could not find an argument to generalize the impossibility to arbitrary values of $n$. However, many party-approval rules (e.g., all Thiele rules and sequential Thiele rules) do not change their outcome when adding voters who approve all parties. For such rules, we can extend \Cref{thm:SAT} to all $n\geq 2k$ by simply adding voters who approve all parties.

On the other side, the assumption that $m\geq k+1$ is crucial for \Cref{thm:SAT}: if $m \leq k$, every rule that constantly returns a fixed committee $W$ with $W(x)\geq 1$ for all $x\in\mathcal{P}$ satisfies the considered axioms. Nevertheless, we can restore the impossibility by strengthening weak representation to weak proportional representation.

\begin{restatable}{theorem}{thmSATtwo}\label{thm:SAT2}
No party-approval rule simultaneously satisfies anonymity, weak proportional representation, and strategy\-proofness if $k\geq 3$, $m\geq 4$, and $2k$ divides $n$.
\end{restatable}

We believe that also the proofs of our results are of interest: for showing \Cref{thm:SAT,thm:SAT2}, we rely on a computer-aided approach called SAT solving. In the realm of social choice, this technique was pioneered by \citet{TaLi09a} and has by now been used to prove a wide variety of results \citep[e.g.,][]{Pete18a,End20a,BBPS21a}. We refer to \citet{GePe17a} for an overview of this technique. 

To apply SAT solving to our problems, we proceed in three steps: first, we encode the problem of finding an anonymous party-approval rule that satisfies strategy\-proofness and weak representation for committees of size $k=3$, $m=4$ parties, and $n=6$ voters as logical formula. By letting a computer program, a so-called SAT solver, show the formula unsatisfiable, we prove the base case of \Cref{thm:SAT,thm:SAT2} for the given parameters. Next, we generalize the impossibility to larger values of $k$, $m$, and $n$ based on inductive arguments. Finally, we verify the computer proof. The following subsections discuss each of these steps in detail.

\begin{remark}
AV satisfies all axioms of \Cref{thm:SAT} except weak representation, and CCAV satisfies all axioms except strategy\-proofness. These examples show that these axioms are required for the impossibility. On the other hand, we could not show that anonymity is necessary for the impossibility and we conjecture that this axiom can be omitted.
\end{remark}

\begin{remark}
For electorates where the committee size $k$ is a multiple of the number of voters $n$, there are voting rules that satisfy weak proportional representation, anonymity, and strategy\-proofness. We can simply let every voter choose $\frac{k}{n}$ parties of the committee independently of the ballots of other voters. This is an important difference to the impossibility by \citet{Pete18a}, which also holds in the case that $n=k$. 
\end{remark}

\begin{remark}
If $k=2$, a variant of AV satisfies all axioms of \Cref{thm:SAT,thm:SAT2}. For introducing this rule $f$, let $\succ$ denote a linear order over the parties and $\mathit{mAV}(A)$ the maximal approval score of a party in the profile $A$. As first step, $f$ removes clones according to $\succ$, i.e., for all parties $x,y$ such that $x\in A_i$ if and only if $y\in A_i$ for all $i\in N$ and $x\succ y$, we remove $y$ from $A$. This results in a reduced profile $A'$. Now, if $\mathit{mAV}(A')\neq \frac{n}{2}$ or there is only a single party with approval score of $\frac{n}{2}$, $f$ assigns both seats to the approval winner. Else, $f$ assigns the seats to the best and second best party with respect to $\succ$ that have an approval score of $\frac{n}{2}$.
\end{remark}

\begin{remark}
\Cref{thm:SAT2} also holds when replacing weak proportional representation with weak representation and Pareto optimality. The reason for this is that for almost all profiles, the combination of weak representation, strategy\-proofness and Pareto optimality implies weak proportional representation. In more detail, this claim holds for all profiles but those where we can split the voters into two sets $N_1$ and $N_2$ such that $A_i=x$ for all $i\in N_1$ and $A_i\in \{\mathcal{P}, \mathcal{P}\setminus \{x\}\}$ for all other voters. As neither our proof nor the inductive arguments requires such profiles, it follows thus that \Cref{thm:SAT2} also holds with the modified set of assumptions. 
\end{remark}

\begin{remark}
We use a significantly stronger strategy\-proofness notion than \citet{Pete18a} for our impossibility theorem. \citet{Pete18a} only considers a rule manipulable if there are profiles $A$, $A'$ such that $A_{-i}=A'_{-i}$, $A_i\cap f(A)\subsetneq A_i\cap f(A')$ and $A_i'\subseteq A_i$. Perhaps surprisingly, our computer program finds a satisfying assignment when using this weaker strategy\-proofness notion and even when dropping the condition on $A_i'$. However, the corresponding party-approval rule seems merely technical and we could not find a closed form representation.
\end{remark}

\subsection{Computer-Aided Theorem Proving}\label{subsec:SAT}

The core observation for computer-aided theorem proving is that for a fixed committee size $k$ and fixed numbers of parties $m$ and voters $n$, there is a very large but finite number of party-approval rules. Hence, we could, at least theoretically, enumerate all rules and check whether they satisfy our requirements. However, the search space grows extremely fast (for $k=3$, $m=4$, and $n=6$, there are roughly $6.2\times 10^{14819544}$ party-approval rules) and we thus use a different idea: we construct a logical formula which is satisfiable if and only if there is an anonymous party-approval rule that satisfies weak representation and strategyproofness for the given parameters of $k$, $m$, and $n$. By showing that the formula is unsatisfiable, we prove \Cref{thm:SAT,thm:SAT2} for fixed parameters. Moreover, we can use computer programs, so-called SAT solvers, to show this.

Subsequently, we explain how we construct the formula and first specify the variables: the idea is to introduce a variable $x_{A,W}$ for each profile $A\in \mathcal{A}^n$ and committee $W\in\mathcal{W}_k$, with the interpretation that $x_{A,W}$ is true if and only if $f(A,k)=W$. However, for this formulation, the mere number of profiles becomes prohibitive when $k=3$, $m=4$, and $n=6$ and we thus apply several optimizations. First, we use anonymity to drastically reduce the number of considered profiles. This axiom states that the order of the voters does not matter for the outcomes and we thus view approval profiles from now on as multisets of approval ballots instead of ordered tuples. Next, we exclude certain approval profiles from the domain by imposing three conditions: \emph{(i)} no voter is allowed to approve all parties, \emph{(ii)} no party can be approved by more than four voters, and \emph{(iii)} the total number of approvals given by all voters does not exceed eleven. We call the domain of all anonymous profiles that satisfy these conditions $\mathcal{A}_\mathit{SAT}^n$. Clearly, if there is no anonymous party-approval rule satisfying strategy\-proofness and weak representation on $\mathcal{A}_\mathit{SAT}^n$, there is also no such function on the full domain $\mathcal{A}^n$. For our last optimization, we note that weak representation requires that a committee $W$ cannot be returned for a profile $A$ if there is a party $x$ with $W(x)=0$ that is uniquely approved by $\frac{n}{k}$ or more voters. Hence, all corresponding variables $x_{A,W}$ must be set to false and we can equivalently omit them. To formalize this, we define $\mathit{WR}(A,k)$ as the set of committees of size $k$ that satisfy weak representation for the profile $A$. Then, we add for every profile $A\in\mathcal{A}_\mathit{SAT}^n$ and every committee $W\in\mathit{WR}(A,k)$ a variable $x_{A,W}$. 

Next, we turn to the constraints of our formula. First, we specify that the formula encodes a function $f$ on $\mathcal{A}_\mathit{SAT}^n$, i.e., for every profile $A\in \mathcal{A}_\mathit{SAT}^n$, there is exactly one committee $W\in\mathit{WR}(A,k)$ such that $x_{A,W}=1$. For this, we add two types of clauses for every profile $A$: the first one specifies that at least one committee is chosen for $A$ and the second one that no more than one committee can be chosen.

\noindent\begin{align*}
    & \bigvee_{\phantom{V,:W}W\in\mathit{WR}(A,k)\phantom{\neq V}}  x_{A,W} & \forall A\in \mathcal{A}_\mathit{SAT}^n
    \\
    & \bigwedge_{V,W\in\mathit{WR}(A,k):V\neq W} \!\!\!\! \lnot x_{A,V}\vee \lnot x_{A,W} \quad & \forall A\in \mathcal{A}_\mathit{SAT}^n%\label{eq:uniqueness}
\end{align*}

Since weak representation and anonymity are encoded in the choice of variables, we only need to add the subsequent constraints for strategy\-proofness. Here, $A^{A_i\rightarrow A_j}$ is the profile derived from $A$ by changing a ballot $A_i$ to $A_j$. 

\noindent\begin{align*}
    \!\!\!\lnot x_{A,V} \lor \lnot x_{A', W} \quad & \forall A,A' \in\mathcal{A}_\mathit{SAT}^n, V\in\mathit{WR}(A,k), \\
    & W\in\mathit{WR}(A',k): \exists A_i, A_j \in \mathcal{A}:\\
    & A'=A^{A_i\rightarrow A_j} \land W(A_i)>V(A_i)
\end{align*}

For committees of size $k=3$, $m=4$ parties, and $n=6$ voters, this construction results in a formula containing $21,418,593$ constraints and a state-of-the-art SAT solver, such as Glucose \citep{audemard2018glucose}, needs less than a minute to prove its unsatisfiability. Our code also provides options which further reduce the size of the formula to speed up the SAT solving (see the supplementary material for details). Consequently, we derive the following result.

\begin{proposition}\label{prop:induction_basis}
There is no party-approval rule that satisfies anonymity, weak representation, and strategy\-proofness if $k=3$, $m=4$, and $n=6$.
\end{proposition}

\subsection{Inductive Arguments}\label{subsec:induction_step}

Since weak propositional representation implies weak representation, \Cref{prop:induction_basis} proves \Cref{thm:SAT,thm:SAT2} for fixed parameters $k$, $m$, and $n$. To complete the proofs of these theorems, we use inductive arguments to generalize the impossibilities to larger parameters and subsequently present them for \Cref{thm:SAT}. For \Cref{thm:SAT2}, only the third claim needs to be adapted and the details can be found in the appendix. 

\begin{restatable}{lemma}{inductionstep}\label{lem:inductionstep}
Assume there is no anonymous party-approval rule $f$ that satisfies weak representation and strategy\-proofness for committees of size $k$, $m$ parties, and $n$ voters. The following claims hold:
\begin{enumerate}[label=(\arabic*), leftmargin=*,topsep=4pt]
    \item For every $\ell\in \mathbb{N}$, there is no such rule for committees of size $k$, $m$ parties, and $\ell \cdot n$ voters.
    \item There is no such rule for committees of size $k$, $m+1$ parties, and $n$ voters.
    \item If $k$ divides $n$, there is no such rule for committees of size $k+1$, $m+1$ parties, and $\frac{n(k+1)}{k}$ voters.
\end{enumerate}
\end{restatable}
\begin{proof}[Proof sketch]
    For all three claims, we prove the contrapositive: if there is an anonymous party-approval rule $f$ that satisfies strategy\-proofness and weak representation for the increased parameters, there is also such a rule $g$ for committees of size $k$, $m$ parties, and $n$ voters. Subsequently, we discuss how to define the rule $g$ for the three different cases: 
    \begin{enumerate}[label=(\arabic*),leftmargin=*,topsep=4pt]
        \item Assume there is $\ell\in\mathbb{N}$ such that $f$ is defined for committees of size $k$, $m$ parties, and $\ell \cdot n$ voters. Given a profile $A$ for $m$ parties and $n$ voters, $g$ copies every voter $\ell$ times to derive the profile $A'$. Then, $g(A,k)=f(A',k)$.
        \item Assume $f$ is defined for committees of size $k$, $m+1$ parties, and $n$ voters. Given a profile $A$ for $m$ parties and $n$ voters, $g$ first constructs the profile $A^{xy}$ by cloning a party $x \in \mathcal P$ into a new party $y \notin \mathcal P$. More formally, $A^{xy}$ is defined by $A_i^{xy}$ if $x\not\in A_i$ and $A_i^{xy}=A_i\cup\{y\}$ otherwise. Finally, $g(A,k,z)=f(A^{xy},k,z)$ for all $z\neq x$ and $g(A,k,x)=f(A^{xy},k,xy)$. 
        \item Assume $k$ divides $n$ and $f$ is defined for committees of size $k+1$, $m+1$ parties, and $\frac{n(k+1)}{k}$ voters. In this case, $g$ maps a profile $A$ for $m$ parties and $n$ voters to the profile $\bar A^{xy}$ defined as follows: first $g$ derives $A^{xy}$ as explained in the previous case and then it adds $\frac{n}{k}$ voters with ballot $xy$. Finally, $g(A,k,z)=f(\bar A^{xy},k+1,z)$ for all $z\neq x$ and $g(A,k,x)=f(\bar A^{xy},k+1,xy)-1$. 
    \end{enumerate}
    
    For all cases, it remains to show that $g$ is a well-defined party-approval rule that satisfies anonymity, weak representation, and strategy\-proofness. Due to space reasons, we explain this only for case (1) and defer the remaining cases to the appendix. In this case, $g$ inherits anonymity from $f$ since permuting the voters in $A$ only permutes the voters in the enlarged profile $A'$. Also, $g$ satisfies weak representation: if $\frac{n}{k}$ or more voters uniquely approve a party $x$ in a profile $A$, at least $\frac{\ell\cdot n}{k}$ voters uniquely approve $x$ in $A'$. Thus, $g(A,x)=f(A',x)\geq 1$ because $f$ satisfies weak representation. Finally, we prove that $g$ is strategy\-proof. Note for this that $f(\bar A, k, \bar A_i)\geq f(\bar A', k, \bar A_i)$ for all profiles $\bar A$, $\bar A'$ that only differ in the ballots of voters who report $\bar A_i$ in $\bar A$. This is true because we can transform $\bar A$ into $\bar A'$ by letting voters with ballot $\bar A_i$ manipulate one after another, and strategy\-proofness shows for every step that the number of seats assigned to parties in $\bar A_i$ cannot increase. Hence, $g$ is strategyproof because if $A$ and $A'$ only differ in a single ballot $A_i$, the enlarged profiles $\bar A$ and $\bar A'$ differ in $\ell$ voters with ballot $A_i$. Thus, $g$ meets all requirements in case (1).
\end{proof}

\subsection{Verification}\label{subsec:verification}

Since \Cref{prop:induction_basis} is proved by automated SAT solving, there is no complete human-readable proof for verifying \Cref{thm:SAT,thm:SAT2}. The standard approach for adressing this issue is to analyze minimal unsatisfiable subsets (MUSes) of the original formula, i.e., subsets of the formula which are unsatisfiable but removing a single constraint makes them satisfiable. Such MUSes are typically much smaller than the original formula, which makes it possible to translate them into a human-readable proof. Unfortunately, this technique does not work for \Cref{prop:induction_basis} because all MUSes that we found (by using the programs haifamuc and muser2 \citep{BeMa12a,NRS14a}) are huge: even after applying several optimizations, the smallest MUS still contained over 20,000 constraints and 635 profiles. Because of the size of the MUSes, any human-readable proof would be unreasonably long and we thus verify our results by other means. Firstly, we have published the code used for proving \Cref{prop:induction_basis} \citep{code}, thus enabling other researchers to reproduce the impossibility. 

Secondly, we provide a human-readable proof for a weakening of \Cref{prop:induction_basis} that additionally uses Pareto optimality. This proof is derived by applying the computer-aided approach explained in \Cref{subsec:SAT} and by analyzing MUSes of the corresponding formula. Hence, it showcases the correctness of our code. Unfortunately, the proof of this weaker claim still takes 11 pages (even though the used MUSes only consist of roughly 500 constraints), and we thus have to defer it to the supplementary material.

Thirdly, we have---analogous to \citet{BBEG16a} and \citet{BSS19b}---verified the correctness of our results with the interactive theorem prover Isabelle/HOL \citep{NPW02a}. Such interactive theorem provers support much more expressive logics and we can hence formalize the entire theorems with all the mathematical notions expressed in a similar way as in \Cref{sec:preliminaries}. For instance, \Cref{fig:Isabelle} displays our Isabelle formalization of weak representation.
\jonasinline{I think the figure would be better readable if we extend it to double-colon width.}
\patrick{True, but that takes more space.}
Our Isabelle/HOL implementation thus directly derives \Cref{prop:induction_basis} as well as \Cref{thm:SAT,thm:SAT2} from the definitions of the axioms. This releases us from the need to check any intermediate steps encoded in Isabelle because Isabelle checks the correctness of these steps for us. Moreover, Isabelle/HOL is highly trustworthy as all proofs have to pass through an inference kernel, which only supports the most basic logical inference steps. Thus, to trust the correctness of our result, one only needs to trust the faithfulness of our Isabelle implementation to the definitions in \Cref{sec:preliminaries}. Such formal proofs are widely considered to be the ``gold standard'' of increasing the trustworthiness of a mathematical result~\citep[e.g.,][]{hales2017formal}. Our formal proof development is available in the \emph{Archive of Formal Proofs}~\cite{ddmil22_afp}.

\begin{figure}
\lstset{
  basicstyle={\footnotesize\ttfamily},
  xleftmargin=0em,
  numbers=none,
  breaklines=true,
  mathescape
}
    \begin{lstlisting}
weak_rep_for_anon_papp_rules n $\mathcal{P}$ k f =
 (anon_PAPP_rule n $\mathcal{P}$ k f $\land$
 ($\forall$A x. anon_papp_profile n $\mathcal{P}$ A $\land$ 
  k * count A {x} $\geq$ n $\rightarrow$ count f(A) x $\geq$ 1))
\end{lstlisting}
    \caption{The Isabelle/HOL code for weak representation. Given the number of voters $n$, the set of parties $\mathcal{P}$, a target committee size $k$, and a function $f$, the code first verifies that $f$ is an anonymous party-approval rule for the given parameters and then requires for every profile $A$ (that is valid for $n$ and $\mathcal{P}$) and every party $x$ that $x$ has at least one seat in $f(A)$ if at least $\frac{n}{k}$ voters uniquely approve $x$.
}
    \label{fig:Isabelle}
\end{figure}

\section{Strategy\-proofness for Unrepresented Voters}\label{sec:wSP}

Since cardinal strategy\-proofness does not allow for attractive party-approval rules, we consider strategy\-proofness for unrepresented voters (\Cref{def:weaksp}) in this section. Instead of prohibiting all voters from manipulating, this property requires that only voters who do not approve any party in the elected committee cannot manipulate. 

As a first result, we prove that CCAV satisfies this axiom and can even be characterized based on strategy\-proofness for unrepresented voters and weak representation within the class of Thiele rules. Hence, CCAV offers an attractive escape route to \Cref{thm:SAT} as it satisfies all axioms of the impossibility when weakening strategy\-proofness. 

\begin{restatable}{theorem}{CCAV}\label{thm:CCAV}
CCAV is the only Thiele rule that satisfies weak representation and strategy\-proofness for unrepresented voters for all committee sizes $k$, numbers of parties $m$, and numbers of voters $n$.
\end{restatable}
\begin{proof}
For proving this theorem, we show that CCAV satisfies the given axioms for all $k$, $m$, and $n$ (Claim 1), and that no other Thiele rule does so (Claim 2).\smallskip

\textbf{Claim 1:} We start by proving that CCAV satisfies weak representation and note for this that \citet{ABC+16a} have shown that CCAV satisfies justified representation in the ABC setting. It thus satisfies weak representation for party-approval elections as this axiom is weaker than justified representation and party-approval elections can be seen as special case of ABC elections. 

Next, we prove by contradiction that CCAV satisfies strategy\-proofness for unrepresented voters. Hence, suppose that there are a voter $i\in N$, profiles $A^1$ and $A^2$, and a committee size $k$ such that $\textit{CCAV}(A^2,k,A_i^1)>\textit{CCAV}(A^1,k,A_i^1)=0$ and $A_j^1=A_j^2$ for all $j\in N\setminus \{i\}$. To simplify notation, let $W^1=\mathit{CCAV}(A^1,k)$ and $W^2=\mathit{CCAV}(A^2,k)$, and define $s(W,A)=|\{i\in N\colon A_i\cap W\neq \emptyset\}|$ as the CCAV-score of a committee $W$ in a profile $A$. Now, the definition of CCAV requires that $s(W^1,A^1)\geq s(W^2, A^1)$ and $s(W^2,A^2)\geq s(W^1, A^2)$. Moreover, since $W^1(A^1_i)=0$ and $A^1_j=A^2_j$ for all voters $j\in N\setminus\{i\}$, it follows that $s(W^1, A^2)\geq s(W^1, A^1)$. Finally, we assumed that $W^2(A^1_i)>0$, which implies that $s(W^2,A^1)\geq s(W^2, A^2)$ since $A^1_j=A^2_j$ for all $j\in N\setminus\{i\}$. By combining these inequalities, we obtain $s(W^2, A^2) \geq s(W^1, A^2) \geq s(W^1, A^1) \geq s(W^2, A^1) \geq s(W^2, A^2)$, which implies that all scores are equal. However, lexicographic tie-breaking implies then that we choose either $W^1$ or $W^2$ for both $A^1$ and $A^2$, which contradicts that $W^1=\mathit{CCAV}(A^1,k)$ and $W^2=\mathit{CCAV}(A^2,k)$.\smallskip

\textbf{Claim 2:} Next, we show that no other Thiele rule but CCAV satisfies weak representation and strategyproofness for unrepresented voters for all $k$, $m$, and $n$. First, observe that AV clearly fails weak representation. Thus, let $f$ be a $w$-Thiele rule other than AV and CCAV. We will show that $f$ fails strategy-proofness for unrepresented voters. Note for this that there is an index $j$ with $w_1 >w_j$ since $f$ is not AV. We denote with $j_0$ the smallest such index, which means that $\forall j < j_0, w_j = w_1 = 1$. If $w_{j_0} = 0$, then $j_0 \ge 3$ because $f$ is not CCAV. Let $\mathcal P = \{a_1, \dots, a_{j_0},b_1, \dots, b_{j_0}\}$ be a set of $m = 2j_0$ parties. We construct the profile $A$ with $n = 2 \cdot \binom{2j_0}{j_0} - 2$ voters and set the target committee size to $k = j_0$. The approval ballots of the voters are defined as follows: voter $1$ reports $\{a_1, \dots, a_{j_0}\}$, voter $2$ reports $\{b_1\}$ and for every set $X\subseteq \mathcal P$ with $|X|=j_0$, $X\neq \{a_1, \dots, a_{j_0}\}$, and $X\neq \{b_1, \dots,b_{j_0}\}$, there are two voters who report $X$ as their ballot. 

First, note that every party appears in exactly $n_c = 2\binom{2j_0-1}{j_0-1}-2$ ballots of the voters $N_c = N\setminus \{1,2\}$. Consequently, every committee $W$ of size $j_0$ gets a total of $\sum_{x\in \mathcal{P}} W(x) |\{i\in N_{c}: x\in A_i\}| = j_0n_c$ approvals from these voters. We use this fact to compute the scores of a committee $W$ derived from these voters. Observe that the committees $W_A=[a_1, \dots, a_{j_0}]$ and $W_B=[b_1,\dots, b_{j_0}]$ receive a score of $j_0n_c$ from the voters in $N_c$ because none of them approves all parties in the committee and $w_1=\dots=w_{j_0-1}=1$. On the other hand, for every other committee $W$, there are at least two voters who approve all parties in $W$. Hence, these voters assign a score of $j_0-1+w_{j_0}$ to the committee. Since the total sum of approvals is constant we derive that the remaining voters in $N_c$ assign at most a score of $j_0(n_c-2)$ to $W$. Hence, the score of $W$ among voters in $N_c$ is upper bounded by $j_0n_c-2(1-w_{j_0})$. 
Finally, if we add the first two voters, $W_A$ obtains a score of $j_0n_c + j_0-1 + w_{j_0}$, $W_B$ of $j_0n_c + 1 < j_0n_c + j_0-1 + w_{j_0}$ (because either $j_0 \ge 3$ or $j_0=2$ and $w_{j_0} > 0$), and the scores of other committees is at most $j_0n_c - 2(1-w_{j_0}) + j_0 <j_0n_c + j_0-1 + w_{j_0}$ (since $w_{j_0}<1$). Hence, $f(A,j_0)=W_A$.

Now, consider the profile $A'$ derived from $A$ by changing the approval ballot of voter $2$ to $\{b_1, \dots, b_{j_0}\}$. Then, the score of the committee $W_A$ does not change and the score of $W_B$ is now equal to the score of $W_A$. Moreover, the same argument as before shows that the score of all other committees is still strictly lower. Hence, committees $W_A$ and $W_B$ are now tied for the win. If the tie-breaking favors $W_B$ over $W_A$, we thus have $f(A', j_0)=W_B$ and voter $2$ can manipulate even though $f(A, j_0, A_2)=0$. Otherwise, we can exchange the roles of $\{a_1,\dots, a_{j_0}\}$ and $\{b_1, \dots, b_{j_0}\}$. Hence, $f$ fails strategy\-proofness for unrepresented voters.
\end{proof}

A natural follow-up question to \Cref{thm:CCAV} is whether party-approval rules other than Thiele rules satisfy strategy\-proofness for unrepresented voters. We partially answer this question by showing that all sequential Thiele rules (but AV) and all divisor methods based on majoritarian portioning (but AV) fail this axiom. Hence, even this weak notion of strategyproofness is a challenging axiom for party-approval elections. We defer the proof of this theorem completely to the appendix; it works by constructing counterexamples similar to Claim 2 in \Cref{thm:CCAV}.

\begin{restatable}{theorem}{wSP}\label{thm:wSP}
All sequential Thiele rules except AV and all divisor methods based on majoritarian portioning except AV fail strategy\-proofness for unrepresented voters for some committee size~$k$, number of parties~$m$, and number of voters~$n$. 
\end{restatable}

\begin{remark}
CCAV becomes highly indecisive if $k\geq m$ since every voter will approve at least one party in the chosen committee. Thus, many seats of the committee will be assigned by the tie-breaking. Hence, CCAV is no attractive rule if $k>m$. Similar arguments show that all $w$-Thiele rules that have an index $j$ with $w_j=0$ are strategy\-proof for unrepresented voters if $k\geq (j-1)m$: in this case, these rules always choose a committee which guarantees every voter $j-1$ representatives and strategy\-proofness for unrepresented voters is trivially satisfied. Consequently, \Cref{thm:CCAV} needs to quantify over the committee size, number of parties, and number of voters.
\end{remark}

\begin{remark}
All results of this section carry over into the ABC setting. For the negative results this follows from the fact that party-approval elections can be seen as a special case of ABC elections (see \Cref{fig:settings}). 
The first claim of \Cref{thm:CCAV} holds since our proof directly translates into the ABC setting.
\end{remark}

\begin{remark}
In much of the literature, committee voting rules are assumed to return sets of committees to avoid tie-breaking \citep[e.g.,][]{LaSk22a}. \Cref{thm:wSP,thm:CCAV} hold also for the irresolute case when using suitable set extensions to compare sets of committees.
\end{remark}

\section{Conclusion}

We study the compatibility of strategy\-proofness and proportional representation for party-approval multiwinner elections, where a multiset of the parties is chosen based on the voters' approval ballots. First, we prove based on a computer-aided approach that strategy\-proofness and minimal notions of proportional representation are incompatible for anonymous party-approval rules. Thus, the incompatibility of strategy\-proofness and proportional representation first observed by \citet{Pete18a} for approval-based committee voting rules (which return sets instead of multisets) also prevails in our more flexible setting. As a second contribution, we investigate a weakening of strategy\-proofness which requires that only voters who do not approve any member of the committee cannot manipulate. Perhaps surprisingly, almost all commonly studied party-approval rules fail even this very weak strategy\-proofness notion. Conversely, we can characterize Chamberlin--Courant approval voting as the unique Thiele rule that satisfies strategy\-proofness for unrepresented voters and a weak representation axiom, thus offering an attractive escape route to our previous impossibility theorem.

Our work offers several directions for future extensions. In particular, we feel that strategy\-proofness for unrepresented voters deserves more attention; for example, we have to leave it open whether weak proportional representation is compatible with this axiom. Furthermore, one can see strategy\-proofness and strategy\-proofness for unrepresented voters as two extreme cases of a parameterization of strategy\-proofness and it thus might be interesting to consider quantified strategy\-proofness notions for party-approval elections.

\section*{Acknowledgments}
We would like to thank Felix Brandt, Markus Brill, Dominik Peters, and Ren\'e Romen for their helpful comments.
Théo Delemazure was supported by the PRAIRIE 3IA Institute under grant ANR-19-P3IA-0001 (e). Tom Demeulemeester was supported by Research Foundation - Flanders under grant \mbox{11J8721N}. Jonas Israel was supported by the Deutsche Forschungsgemeinschaft under grant \mbox{BR~4744/2-1}. Patrick Lederer was supported by the Deutsche Forschungsgemeinschaft under grant \mbox{BR 2312/12-1}.

\clearpage

\appendix

\section{Appendix}
\subsection{Optimizations for SAT Solving}\label{appendix:SAToptimizations}

As mentioned in \Cref{subsec:SAT}, our code offers options (\texttt{--cleverWR} and \texttt{--SymmetryBreaking}) that reduce the size of the constructed formula and thus speed up the SAT solving. Subsequently, we explain these options and their correctness.

We start by discussing the option \texttt{--cleverWR}. The idea of this option is that the interplay of weak representation and strategy\-proofness can be used to restrict the set of feasible committees for some profiles even further. In particular, \texttt{--cleverWR} encodes the following lemma.

\begin{restatable}{lemma}{strongWeakRep}\label{lem:strongWeakRep}	
Let $f$ denote an anonymous party-approval rule that satisfies strategy\-proofness and weak representation and let $k$ denote the target committee size. Moreover, we consider a profile $A$ and define $X=\{x\in \mathcal{P}: \exists G\subseteq N: |G|\geq \frac{n}{k} \land \forall i\in G: A_i={x}\}$ as the set of parties that are each uniquely approved by at least $\frac{n}{k}$ voters. If there is a set of voters $G=\{i_1,\dots i_\ell\}\subseteq N$ such that $|G|\geq \frac{n}{k}$, $A_{i_1}\subseteq A_{i_2}\subseteq\cdots\subseteq A_{i_\ell}$, and $A_{i_1}\not\subseteq X$, then $f(A, k, A_{i_\ell})\geq |A_{i_\ell}\cap X|+1$.
\end{restatable}
\begin{proof}
Let $f$ denote an anonymous party-approval rule that satisfies strategy\-proofness and weak representation, and consider a profile $A$ and a set of voters $G=\{i_1,\dots, i_\ell\}$ as defined in the lemma. Moreover, let $X$ denote the set of parties that are uniquely approved by at least $\frac{n}{k}$ voters in $A$. Weak representation requires that $f(A,x)\geq 1$ for all parties $x\in X$. For the sake of contradiction, we suppose that $f(A, A_{i_\ell})\leq |A_{i_\ell}\cap X|$. Since every party in $X$ is chosen at least once, it follows that $f(A, A_{i_\ell})= |A_{i_\ell}\cap X|$. Furthermore, by assumption, there is a party $y\in A_{i_1}\subseteq A_{i_2}\subseteq \dots\subseteq A_{i_\ell}$ with $y\not\in X$. In particular, no voter in $G$ approves a party $x\in X$ uniquely, i.e., the parties in $A_{i_\ell}\cap X$ will be assigned at least one seat independent of the ballots of the voters in $G$. 

Next, we consider the sequence of profiles $A^{\ell+1},\cdots, A^1$ such that $A^{\ell+1}=A$ and $A^{j}$ is derived from $A^{j+1}$ by letting voter $i_j\in G$ change her ballot to $y$. Now consider a single step from $A^{j+1}$ to $A^j$ and suppose that $f(A^{j+1}, A_{i_j}^{j+1})\leq |A_{i_j}^{j+1}\cap X|$. By our previous observation, every party $x\in A_{i_j}\cap X$ is chosen at least once and the above inequality is tight. Moreover, strategy\-proofness requires that $|A_{i_j}^{j+1}\cap X|= f(A^{j+1}, A_{i_j}^{j+1})\geq f(A^{j}, A_{i_j}^{j+1})$. Since each party $x\in A_{i_j}^{j+1}\cap X$ must also be chosen at least once for $A^j$, it follows that each of these parties is assigned exactly one seat in $f(A^j)$. Consequently, no party in $A^{j+1}_{i_j}\setminus X$ can be elected for $A^j$ because of strategy\-proofness. Since $A^j_{j-1}=A_{j-1}\subseteq A_j$, we thus derive that $f(A^j, A^j_{j-1})\leq |A^j_{j-1}\cap X|$. Finally, we can now repeat the same argument for $A^{j-1}$. Since we assume that $f(A, A_{i_\ell})=f(A^{\ell+1}, A_{i_\ell})\leq |A_{i_\ell}\cap X|$, we infer inductively that $f(A^1,A_{i_1})\leq |A_{i_1}\cap X|$. In particular, this means that still no party $x\in A_{i_1}\setminus X$ is chosen. However, all voters in $G$ now approve only $y$ and weak representation thus requires that this party is chosen. This is a contradiction, which shows that $f(A, A_{i_\ell})\geq |A_{i_\ell}\cap X|$.
\end{proof}

For instance, in the subsequent profile $A$, this lemma requires that one seat is assigned to $a$ or $b$, and one seat is assigned to $c$ or $d$ when choosing committees of size $k=3$. We encode \Cref{lem:strongWeakRep} just as weak representation: we omit the variables $x_{A,W}$ if the committee $W$ fails the lemma on the profile $A$ (e.g., for the profile $A$ shown below, we do not introduce a variable for the committee $[a,a,a]$). Or, more formally, when the option \texttt{--cleverWR} is used, then we only introduce variables for all profiles $A\in \mathcal{A}_\mathit{SAT}$ and committees $W\in \mathit{WR}(A,k)$ such that \Cref{lem:strongWeakRep} is satisfied. 

\begin{center}
   \begin{tabular}{cccccccc}
     $A$: & $a$ & $ab$ & $b$ & $c$ & $cd$ & $d$
\end{tabular} 
\end{center}

Next, we turn to the option \texttt{--SymmetryBreaking}. The idea for this optimization is that some profiles are symmetric, which allows us to treat outcomes analogously. For instance, in the profile $A$ shown above, the committee $[a,b,c]$ is symmetric to the committee $[a,b,d]$. Hence, if there is a party-approval rule satisfying anonymity, weak representation, strategy\-proofness, and $f(A,k)=[a,b,c]$, we can construct another rule $f'$ satisfying the same axioms and $f'(A,k)=[a,b,d]$. This is formalized by the next lemma. 

\begin{restatable}{lemma}{symmetrybreaking}\label{lem:symmetrybreaking}	
Let $f$ be an anonymous party-approval rule that satisfies strategy\-proofness and weak representation, and let $\tau:\mathcal P\to \mathcal P$ be a permutation of the parties. The party-approval rule $f^\tau(A,k,x)=f(\tau(A),k,\tau(x))$, where $\tau(A)$ denotes the profile such that $\tau(x)\in \tau(A)_i$ if and only if $x\in A_i$ for all $i\in N$ and $x\in \mathcal{P}$, satisfies the same axioms as $f$. 
\end{restatable}
\begin{proof}
Let $f$ denote a party-approval rule that satisfies the given axioms, and let $\tau:P\mapsto P$ denote an arbitrary permutation on the parties. Moreover, define the party-approval rule $f^\tau$ as in the lemma. Clearly, $f^\tau$ inherits anonymity from $f$. Next, $f^\tau$ satisfies weak representation: if a party $x$ is uniquely approved by at least $\frac{n}{k}$ voters in $A$, then $\tau(x)$ is uniquely approved by at least $\frac{n}{k}$ voters in $\tau(A)$. Since $f$ satisfies weak representation, $f^\tau(A,k,x)=f(\tau(A),k,\tau(x))\geq1$ which proves that $f^\tau$ also satisfies this axiom. Finally, $f^\tau$ inherits strategy\-proofness from $f$; otherwise there are two profiles $A$ and $A'$ and a voter $i\in N$ such that $A_j=A'_j$ for all $j\in N\setminus \{i\}$ and $f^\tau(A', k, A_i)>f^\tau (A,k,  A_i)$. However, this implies that $f(\tau(A'), k, \tau(A)_i)>f(\tau(A), k, \tau(A)_i)$, proving that $f$ fails strategy\-proofness, too. Since this contradicts our initial assumption, $f^\tau$ must be strategy\-proof and satisfies thus all required axioms.
\end{proof}

As consequence of \Cref{lem:symmetrybreaking}, we can remove for a single profile all symmetric outcomes but one, and we choose the profile $A$ for this. In particular, note here that \Cref{lem:strongWeakRep} requires for $A$ that at least one seat is assigned to $a$ or $b$ and that at least one seat is assigned to $c$ or $d$. Moreover, all committees satisfying this condition are symmetric to either $[a,a,c]$ or $[a,b,c]$. Or, in other words, if we show that no party-approval rule satisfies anonymity, weak representation, strategy\-proofness, and $f(A,3)\in \{[a,a,c], [a,b,c]\}$, then \Cref{lem:strongWeakRep,lem:symmetrybreaking} imply that no such rule exists in general. Consequently, the option \texttt{--SymmetryBreaking} exactly adds the constraint that either $[a,a,c]$ or $[a,b,c]$ must be chosen for the profile $A$.

\subsection{Proofs of the Inductive Arguments}\label{appendix:inductive_arguments}

In this section, we complete the proof of \Cref{lem:inductionstep} and prove also the inductive arguments required for \Cref{thm:SAT2}. For proving these results, we first discuss an auxiliary claim.

\begin{fact}\label{fact2}
    Assume that $f$ is a party-approval rule that satisfies strategy\-proofness and weak representation. For every subset of parties $X \subseteq \mathcal P$, profile $A$, committee size $k$, and integer $\ell \leq k$, it holds that $f(A,X)\geq \ell$ if $|X|\geq \ell$ and at least $\ell\cdot \lceil\frac{n}{k}\rceil$ voters report $X$ as their approval ballot.
\end{fact}
\begin{proof}
Consider arbitrary values for $k$, $m$, and $n$ and assume that $f$ is a party-approval rule satisfying strategy\-proofness and weak representation. We prove the statement by induction on $\ell$ and first consider the case that $\ell=1$. Thus, let $A$ denote a profile and $X$ a non-empty set of parties such that $A_i=X$ for at least $\lceil\frac{n}{k}\rceil$ voters. Moreover, consider the profile $A'$ derived from $A$ by letting all voters who report $X$ change their preference to approving a single party $x\in X$. It follows from weak representation that $f(A',k,x)\geq 1$ and thus $f(A',k, X)\geq 1$. Then, a repeated application of strategy\-proofness shows that $f(A,k,X)\geq f(A',k,X)\geq 1$. In more detail, we can let the voters approving $X$ change their preference one after another and strategy\-proofness states for every step that the number of seats assigned to parties of $X$ cannot increase. This proves the induction basis.

For the induction step, we assume that the fact holds for all profiles $A$, sets of parties $X$, and integers $1,\dots, \ell$ and we prove the claim for $\ell+1$. Hence, consider a profile $A$, a set of voters $G$, and a set of parties $X$ such that $|X|\geq \ell+1$, $|G|\geq (\ell+1)\cdot \lceil\frac{n}{k}\rceil$, and $A_i=G$ for all $i\in G$. Next, we derive profile $A'$ from $A$ by letting $\lceil\frac{n}{k} \rceil$ voters of $G$ report a single party $x\in X$ and the remaining voters in $G$ report $A\setminus \{x\}$. Weak representation requires for $A'$ that $x$ needs to be chosen at least once. Moreover, the induction hypothesis implies that $f(A',k, X\setminus \{x\})\geq \ell$ because $|X\setminus \{x\}|\geq \ell$ and at least $\ell \cdot \lceil\frac{n}{k}\rceil$ voters report $X\setminus \{x\}$. Hence, $f(A',k,X)\geq \ell+1$ and a repeated application of strategy\-proofness shows that $f(A,k,X)\geq f(A',k,X)$. This proves the lemma. 
\end{proof}

Based on Fact \ref{fact2}, we show next \Cref{lem:inductionstep}.

\inductionstep*
\begin{proof}
The proof of claim (1) is in the main body and we thus focus on the cases (2) and (3). For both cases we show the contrapositive, and we focus first on the case that there is an anonymous party-approval rule $f$ for committees of size $k$, $m+1$ parties, and $n$ voters that satisfies strategy\-proofness and weak representation. Based on $f$, we define another rule $g$ for committees of size $k$, $m$ parties, and $n$ voters as follows: given a profile $A$ for $m$ parties and $n$ voters, $g$ derives a profile $A^{xy}$ for $m+1$ parties by cloning a party $x$. Formally, $A^{xy}$ is defined by $A_i^{xy}$ if $x\not\in A_i$ and $A_i^{xy}=A_i\cup\{y\}$ else (where $y$ is a new party). Then, we set $g(A,k,z)=f(A^{xy},k,z)$ for $z\neq x$ and $g(A^{xy},k,x)=f(A,k,xy)$. 

Subsequently, we prove that $g$ is a party-approval rule that satisfies all required axioms. First, $g$ is well-defined for committees of size $k$ as it assigns exactly $k$ seats to the $m$ parties given as input. Moreover, it clearly inherits anonymity from $f$. Thirdly, $g$ satisfies strategy\-proofness because otherwise $f$ would be manipulable, too. In more detail, assume that a voter $i$ can manipulate $g$ at a profile $A$. Since $y\in A_i^{xy}$ if and only if $x\in A_i^{xy}$, $g(A,k,x)=f(A,k,xy)$, and $g(A,k,z)=f(A,k,z)$, voter $i$ can manipulate $f$ at $A^{xy}$. Or, in other words, if $f$ is strategy\-proof, $g$ is, too. Finally, $g$ also satisfies weak representation, which follows by a case distinction on the parties. If a party $z\neq x$ is uniquely approved by $\lceil \frac{n}{k}\rceil$ voters in the original profile $A$, then it is also uniquely approved by these voters in $A^{xy}$. Hence, $g(A,k,z)=f(A^{xy},k,z)\geq 1$ due to the weak representation of $f$. On the other hand, if $x$ is uniquely approved by at least $\lceil \frac{n}{k}\rceil$ voters, these voters approve $xy$ in $A^{xy}$. Then, \Cref{fact2} shows that $g(A, k, x)=f(A^{xy}, k ,xy)\geq 1$. This shows that $g$ satisfies all axioms and thus proves the second claim.

For the third case, suppose that $k$ divides $n$ and that $f$ is defined for committees of size $k+1$, $m+1$ parties, and $\frac{n(k+1)}{k}$ voters. 
In this case, $g$ maps a profile $A$ for $m$ parties and $n$ voters to the profile $\bar A^{xy}$ defined as follows: first $g$ derives $A^{xy}$ as explained in the first paragraph and then it adds $\frac{n}{k}$ voters with ballot $xy$. Finally, $g(A,k,z)=f(\bar A^{xy},k+1,z)$ for all $z\neq x$ and $g(A,k,x)=f(\bar A^{xy},k+1,xy)-1$. Note that $g$ is a well-defined party-approval rule for committees of size $k$ since \Cref{fact2} ensures that $f(\bar A^{xy}, k+1, xy)\geq 1$ for all profiles $A$. Thus, $g(A,k,z)\geq 0$ for all parties $z$. 

Next, we show that $g$ satisfies all required axioms and note first that it clearly inherits anonymity from $f$. Moreover, an analogous argument as in the last case shows that $g$ is strategy\-proof. Finally, we prove that $g$ satisfies weak representation. Note for this that if a party $z\neq x$ is uniquely approved by at least $\frac{n}{k}$ voters in a profile $A$, it is also uniquely approved by these voters in $\bar A^{xy}$. Since $\frac{n(k+1)}{k(k+1)}=\frac{n}{k}$, it thus follows that $f(\bar A^{xy},k+1,z)\geq 1$ due to the weak representation of $f$ and, by definition, we infer $g(A,k,z)\geq 1$. On the other hand, if $x$ is uniquely approved by at least $\frac{n}{k}$ voters in $A$, there are at least $2\frac{n}{k}$ voters in $\bar A^{xy}$ that report $xy$ in $\bar A^{xy}$. Since $\frac{n(k+1)}{k(k+1)}=\frac{n}{k}$, we infer from \Cref{fact2} that $f(\bar A^{xy}, k+1,xy)\geq 2$, which implies that $g(A,k,x)\geq 1$. Hence, $g$ satisfies all required axioms, which proves also the last case.
\end{proof}

Finally, we discuss the inductive argument required for \Cref{thm:SAT2}. Since weak proportional representation implies weak representation, we can use the first two cases of \Cref{lem:inductionstep} to increase the numbers of voters and parties. Hence, we subsequently explain why we can increasing the committee size without increasing the number of parties. By first applying this lemma to get the desired committee size and then using the two inductive arguments of \Cref{lem:inductionstep}, \Cref{thm:SAT2} follows from \Cref{prop:induction_basis}.

\begin{lemma}\label{lem:IS_committee_wpr}
Assume that $n, k$ are integers such that $n$ is a multiple of $k$.
If there is no party-approval rule that satisfies anonymity, strategy\-proofness, and weak proportional representation for committees of size $k$, $m$ parties, and $n$ voters, there is also no such rule for committees of size $k+1$, $m$ parties, and $\frac{n(k+1)}{k}$ voters. 
\end{lemma}
\begin{proof}
We prove the contrapositive and thus assume that there is a party-approval rule $f$ that satisfies anonymity, strategy\-proofness, and weak proportional representation for committees of the size $k+1$, $m$ parties, and $\frac{n(k+1)}{k}$ voters. Based on $f$, we construct a party-approval rule $g$ for committees of size $k$, $m$ parties, and $n$ voters that satisfies the same axioms. The rule $g$ is defined as follows: given a profile $A$, $g$ appends $\frac{n}{k}$ voters who only approve a party $x$ to derive a new profile $A'$. Finally, $g(A,k,z)=f(A',k+1,z)$ for all $z\neq x$ and $g(A,k,x)=f(A',k+1,x)-1$. 

It remains to show that $g$ satisfies all requirements. First, note that $g$ is a well-defined party-approval rule for $m$, $n$ and $k$ as $g(A,k,z)\geq 0$ for all $z\in \mathcal{P}$. In particular, $g(A,k,x)=f(A',k+1,x)-1\geq 0$ because weak proportional representation requires that $f(A',k+1,x)\geq 1$ since at least $\frac{n}{k}= \frac{n(k+1)}{k(k+1)}$ voters approve $x$ uniquely in $A'$. Next, $g$ clearly inherits anonymity from $f$. As third point, note that $g$ satisfies strategy\-proofness because for all profiles $A$, $g(A,k)$ differs from $f(A',k+1)$ only in the fact that $f$ assigns one more seat to $x$. Hence, if a voter could manipulate $g$, he can manipulate $f$ at the corresponding profile $A'$, contradicting the strategy\-proofness of $f$. Finally, we show that $g$ satisfies weak proportional representation and consider for this a profile $A$ and a party $z$ such that $j\geq \frac{\ell n}{k}$ voters uniquely approve $z$. If $z\neq x$, the same $j$ voters approve $z$ uniquely in $A'$ and the weak proportional representation requires of $f$ entails that $g(A, k, z)=f(A',k+1,z)\geq \ell$ because $j\geq \frac{\ell n}{k}=\frac{\ell n(k+1)}{k(k+1)}$. On the other hand, if $z=x$, there are $j+\frac{n}{k}\geq \frac{(\ell+1)n}{k}$ voters who approve $x$ uniquely in $A'$. Hence, we derive again that $g(A, k, x)=f(A', k+1,x)-1\geq \ell$, which shows that $g$ satisfies all required axioms.
\end{proof}

\subsection{Proof of a Weakened Variant of \Cref{prop:induction_basis}}\label{appendix:verification}
\newcommand{\PO}[2]{$x_{#2}^{#1}=0$ (PO)\\}
\newcommand{\WR}[2]{$x_{#2}^{#1}\geq 1$ (WR)\\}
\newcommand{\SPl}[4]{$x_{#3}^{#1}\leq x_{#3}^{#2}=#4$ (SP from $A^{#2}$ to $A^{#1}$)\\}
\newcommand{\SPg}[4]{$x_{#3}^{#1}\geq x_{#3}^{#2}=#4$ (SP from $A^{#1}$ to $A^{#2}$)\\}
\newcommand{\aux}[3]{$x_{#2}^{#1}\geq #3$ (\Cref{lem:strongWeakRep})\\}

As promised in \Cref{subsec:verification}, we discuss here a human-readable proof of a weakened variant of \Cref{prop:induction_basis} that additionally uses Pareto optimality. We derive this proof from the computer-aided approach explained in \Cref{subsec:SAT} with two minor modifications. Firstly, we encode Pareto optimality based on the same idea as weak representation: we omit the variables for profiles $A\in \mathcal{A}_\mathit{SAT}^n$ and committees $W\in \mathit{WR}(A,k)$ such that $W$ fails Pareto optimality in $A$. Secondly, we use the optimizations described in \Cref{appendix:SAToptimizations} to make the proof as short as possible (note here that \Cref{lem:symmetrybreaking} also preserves Pareto optimality). The proof shown below is then derived by investigating the MUSes of this formula. Since the modifications are minor, this human-readable proof showcases the correctness of our computer program. 

\begin{restatable}{proposition}{inductionbasisweak}\label{cor:induction_basis_weak}
There is no anonymous party-approval rule that satisfies weak representation, Pareto optimality, and strategy\-proofness if $k=3$, $m=4$, and $n=6$.
\end{restatable}
\begin{proof}
Since we are only interested in committees of size $3$ in the subsequent argument, we omit from now on the committee size in our notation, i.e., $f(A)$ stands for $f(A,3)$. Now, assume for contradiction that there is an anonymous party-approval rule $g$ that satisfies strategy\-proofness, Pareto optimality, and weak representation for $k=3$, $m=4$, and $n=6$. Using \Cref{lem:symmetrybreaking}, it follows that there is also a rule $f$ that satisfies these axioms and that $f(A^1)=[a,a,b]$ or $f(A^1)=[a,b,c]$ for the profile $A^1$ shown below. 
\begin{center}
\begin{tabular}{ccccccc}
    $A^1$ & $a$ & $b$ & $ab$ & $c$ & $d$ & $cd$
\end{tabular}
\end{center}

Next, we consider the profile $A^2$, in which $a$ Pareto dominates $b$. Hence, Pareto optimality entails that $f(A^2,b)=0$ and strategy\-proofness requires that $f(A^2,d)=0$; otherwise, the voter with preference $abc$ can manipulate by deviating to $A^1$. Finally, \Cref{lem:strongWeakRep} implies that $f(A^2,ab)\geq 1$ and $f(A^2,cd)\geq 1$, so $f(A^2)\in \{[a,a,c], [a,c,c]\}$. 
\begin{center}
\begin{tabular}{ccccccc}
    $A^2$ & $a$ & $ab$ & $c$ & $abc$ & $d$ & $cd$
\end{tabular}
\end{center}

For deriving an impossibility from this point on, we rely on the profiles shown in \Cref{tab:profiles} and several case distinctions. Also, we use a more short-hand notation: for every set $X$ and profile $A^k$, we define $x_X^k=f(A^k, X)$. Furthermore, we write down our derivations in table form to keep the proof as short as possible. Each row of these tables contains a profile, the possibly chosen committees for this profile, and explanations of why no other committee is feasible. Since we derive these proofs from our computer program, there are exactly four different reasons to why a committee is not possible: Pareto optimality (PO), weak representation (WR), \Cref{lem:strongWeakRep}, and strategy\-proofness (SP). 

As example, we explain this notation for the first steps in Case 1.1. This case starts by making two assumptions: firstly, we use the assumption of Case 1 that $f(A^2)=[a,c,c]$ and secondly, we suppose for contradiction that $f(A^{18})=[c,c,c]$. Next, we derive $f(A^{10})$ based on strategy\-proofness: if $f(A^{10},c)<3$, then voter $4$ can manipulate by deviating to $A^{18}$, i.e., strategy\-proofness from $A^{10}$ to $A^{18}$ requires that $f(A^{10},c)\geq f(A^{18},c)=3$. As third step, we derive the possible committees for $A^{29}$ by applying several axioms: \emph{(i)} Pareto optimality requires that $f(A^{29},b)=0$; \emph{(ii)} weak representation states that $f(A^{29},c)\geq 1$; \emph{(iii)} \Cref{lem:strongWeakRep} requires that $f(A^{29}, ad)\geq 1$; and \emph{(iv)} strategy\-proofness from $A^{10}$ to $A^{29}$ requires that $f(A^{29},a)=0$. For the last point, note that the first voter could manipulate by deviating from $A^{10}$ to $A^{29}$ if $f(A^{29},a)\geq 1$. Consequently, $f(A^{29})\in \{[c,c,d], [c,d,d]\}$. The proof continues with such steps until it derives that no feasible outcome remains for $A^{13}$. This contradicts the definition of a party-approval rule and thus shows that the assumption $f(A^{18})=[c,c,c]$ is wrong. 
\newpage

\newcommand{\R}[3]{\IfEqCase{#1}{%
{1}{$A^{1}$:			& $a$	& $b$	& $ab$	& $c$	& $d$	& $cd$	& #2 & \makecell[t]{#3} 	\\}
{2}{$A^{2}$:			& $a$	& $ab$	& $c$	& $abc$	& $d$	& $cd$	& #2 & \makecell[t]{#3} 	\\}
{3}{$A^{3}$:			& $a$	& $b$	& $c$	& $bc$	& $abc$	& $d$	& #2 & \makecell[t]{#3} 	\\}
{4}{$A^{4}$:			& $a$	& $b$	& $c$	& $abc$	& $d$	& $cd$	& #2 & \makecell[t]{#3} 	\\}
{5}{$A^{5}$:			& $a$	& $ab$	& $c$	& $c$	& $c$	& $ad$	& #2 & \makecell[t]{#3} 	\\}
{6}{$A^{6}$:			& $a$	& $ab$	& $c$	& $c$	& $abc$	& $d$	& #2 & \makecell[t]{#3} 	\\}
{7}{$A^{7}$:			& $a$	& $ab$	& $c$	& $c$	& $d$	& $ad$	& #2 & \makecell[t]{#3} 	\\}
{8}{$A^{8}$:			& $a$	& $ab$	& $c$	& $c$	& $d$	& $cd$	& #2 & \makecell[t]{#3} 	\\}
{9}{$A^{9}$:			& $a$	& $ab$	& $c$	& $c$	& $ad$	& $cd$	& #2 & \makecell[t]{#3} 	\\}
{10}{$A^{10}$:			& $a$	& $c$	& $c$	& $c$	& $abc$	& $d$	& #2 & \makecell[t]{#3} 	\\}
{11}{$A^{11}$:			& $a$	& $c$	& $c$	& $c$	& $d$	& $ad$	& #2 & \makecell[t]{#3} 	\\}
{12}{$A^{12}$:			& $a$	& $c$	& $c$	& $c$	& $d$	& $cd$	& #2 & \makecell[t]{#3} 	\\}
{13}{$A^{13}$:			& $a$	& $c$	& $c$	& $c$	& $ad$	& $ad$	& #2 & \makecell[t]{#3} 	\\}
{14}{$A^{14}$:			& $a$	& $c$	& $c$	& $c$	& $ad$	& $cd$	& #2 & \makecell[t]{#3} 	\\}
{15}{$A^{15}$:			& $a$	& $c$	& $c$	& $bc$	& $abc$	& $d$	& #2 & \makecell[t]{#3} 	\\}
{16}{$A^{16}$:			& $a$	& $c$	& $c$	& $abc$	& $d$	& $d$	& #2 & \makecell[t]{#3} 	\\}
{17}{$A^{17}$:			& $a$	& $c$	& $c$	& $abc$	& $d$	& $ad$	& #2 & \makecell[t]{#3} 	\\}
{18}{$A^{18}$:			& $a$	& $c$	& $c$	& $abc$	& $d$	& $cd$	& #2 & \makecell[t]{#3} 	\\}
{19}{$A^{19}$:			& $a$	& $c$	& $bc$	& $bc$	& $abc$	& $d$	& #2 & \makecell[t]{#3} 	\\}
{20}{$A^{20}$:			& $b$	& $c$	& $c$	& $c$	& $abc$	& $d$	& #2 & \makecell[t]{#3} 	\\}
{21}{$A^{21}$:			& $b$	& $c$	& $c$	& $c$	& $d$	& $cd$	& #2 & \makecell[t]{#3} 	\\}
{22}{$A^{22}$:			& $b$	& $c$	& $c$	& $c$	& $ad$	& $bd$	& #2 & \makecell[t]{#3} 	\\}
{23}{$A^{23}$:			& $b$	& $c$	& $c$	& $c$	& $bd$	& $cd$	& #2 & \makecell[t]{#3} 	\\}
{24}{$A^{24}$:			& $b$	& $c$	& $c$	& $ac$	& $abc$	& $d$	& #2 & \makecell[t]{#3} 	\\}
{25}{$A^{25}$:			& $b$	& $c$	& $c$	& $ac$	& $d$	& $cd$	& #2 & \makecell[t]{#3} 	\\}
{26}{$A^{26}$:			& $b$	& $c$	& $ac$	& $abc$	& $d$	& $cd$	& #2 & \makecell[t]{#3} 	\\}
{27}{$A^{27}$:			& $ab$	& $c$	& $c$	& $c$	& $ad$	& $bd$	& #2 & \makecell[t]{#3} 	\\}
{28}{$A^{28}$:			& $ab$	& $c$	& $c$	& $c$	& $ad$	& $abd$	& #2 & \makecell[t]{#3} 	\\}
{29}{$A^{29}$:			& $c$	& $c$	& $c$	& $abc$	& $d$	& $ad$	& #2 & \makecell[t]{#3} 	\\}
{30}{$A^{30}$:			& $c$	& $c$	& $c$	& $abc$	& $d$	& $bd$	& #2 & \makecell[t]{#3} 	\\}
{31}{$A^{31}$:			& $c$	& $c$	& $c$	& $abc$	& $ad$	& $bd$	& #2 & \makecell[t]{#3}	    \\}
{32}{$A^{32}$:			& $a$	& $a$	& $a$	& $c$	& $abc$	& $d$	& #2 & \makecell[t]{#3}  	\\}
{33}{$A^{33}$:			& $a$	& $a$	& $a$	& $c$	& $d$	& $ad$	& #2 & \makecell[t]{#3}  	\\}
{34}{$A^{34}$:			& $a$	& $a$	& $a$	& $c$	& $d$	& $cd$	& #2 & \makecell[t]{#3}  	\\}
{35}{$A^{35}$:			& $a$	& $a$	& $a$	& $c$	& $ad$	& $cd$	& #2 & \makecell[t]{#3}  	\\}
{36}{$A^{36}$:			& $a$	& $a$	& $a$	& $abc$	& $d$	& $cd$	& #2 & \makecell[t]{#3}  	\\}
{37}{$A^{37}$:			& $a$	& $a$	& $a$	& $d$	& $d$	& $bd$	& #2 & \makecell[t]{#3}  	\\}
{38}{$A^{38}$:			& $a$	& $a$	& $a$	& $d$	& $d$	& $cd$	& #2 & \makecell[t]{#3}  	\\}
{39}{$A^{39}$:			& $a$	& $a$	& $a$	& $d$	& $bd$	& $bd$	& #2 & \makecell[t]{#3}  	\\}
{40}{$A^{40}$:			& $a$	& $a$	& $b$	& $b$	& $d$	& $ad$	& #2 & \makecell[t]{#3}  	\\}
{41}{$A^{41}$:			& $a$	& $a$	& $b$	& $ab$	& $d$	& $d$	& #2 & \makecell[t]{#3}  	\\}
{42}{$A^{42}$:			& $a$	& $a$	& $b$	& $c$	& $abc$	& $d$	& #2 & \makecell[t]{#3}  	\\}
{43}{$A^{43}$:			& $a$	& $a$	& $b$	& $c$	& $d$	& $bcd$	& #2 & \makecell[t]{#3}  	\\}
{44}{$A^{44}$:			& $a$	& $a$	& $b$	& $ac$	& $abc$	& $d$	& #2 & \makecell[t]{#3}  	\\}
{45}{$A^{45}$:			& $a$	& $a$	& $b$	& $ac$	& $d$	& $ad$	& #2 & \makecell[t]{#3}  	\\}
{46}{$A^{46}$:			& $a$	& $a$	& $b$	& $ac$	& $d$	& $bd$	& #2 & \makecell[t]{#3}  	\\}
{47}{$A^{47}$:			& $a$	& $a$	& $b$	& $ac$	& $ad$	& $bd$	& #2 & \makecell[t]{#3}  	\\}
{48}{$A^{48}$:			& $a$	& $a$	& $b$	& $abc$	& $d$	& $ad$	& #2 & \makecell[t]{#3}  	\\}
{49}{$A^{49}$:			& $a$	& $a$	& $b$	& $d$	& $d$	& $ad$	& #2 & \makecell[t]{#3}  	\\}
{50}{$A^{50}$:			& $a$	& $a$	& $b$	& $d$	& $d$	& $bd$	& #2 & \makecell[t]{#3}  	\\}
{51}{$A^{51}$:			& $a$	& $a$	& $b$	& $d$	& $d$	& $bcd$	& #2 & \makecell[t]{#3}  	\\}
{52}{$A^{52}$:			& $a$	& $a$	& $b$	& $d$	& $ad$	& $bd$	& #2 & \makecell[t]{#3}  	\\}
{53}{$A^{53}$:			& $a$	& $a$	& $b$	& $d$	& $ad$	& $bcd$	& #2 & \makecell[t]{#3}  	\\}
{54}{$A^{54}$:			& $a$	& $a$	& $b$	& $d$	& $cd$	& $bcd$	& #2 & \makecell[t]{#3}  	\\}
{55}{$A^{55}$:			& $a$	& $a$	& $ab$	& $c$	& $abc$	& $d$	& #2 & \makecell[t]{#3}  	\\}
{56}{$A^{56}$:			& $a$	& $a$	& $ab$	& $d$	& $d$	& $bd$	& #2 & \makecell[t]{#3}  	\\}
{57}{$A^{57}$:			& $a$	& $a$	& $ab$	& $d$	& $d$	& $cd$	& #2 & \makecell[t]{#3}  	\\}
{58}{$A^{58}$:			& $a$	& $a$	& $ab$	& $d$	& $d$	& $bcd$	& #2 & \makecell[t]{#3}  	\\}
{59}{$A^{59}$:			& $a$	& $a$	& $c$	& $c$	& $d$	& $ad$	& #2 & \makecell[t]{#3}  	\\}
{60}{$A^{60}$:			& $a$	& $a$	& $c$	& $c$	& $d$	& $abd$	& #2 & \makecell[t]{#3}  	\\}
{61}{$A^{61}$:			& $a$	& $a$	& $c$	& $abc$	& $d$	& $d$	& #2 & \makecell[t]{#3}  	\\}
{62}{$A^{62}$:			& $a$	& $a$	& $c$	& $abc$	& $d$	& $ad$	& #2 & \makecell[t]{#3}  	\\}
{63}{$A^{63}$:			& $a$	& $a$	& $c$	& $abc$	& $d$	& $abd$	& #2 & \makecell[t]{#3}  	\\}
{64}{$A^{64}$:			& $a$	& $a$	& $c$	& $abc$	& $d$	& $cd$	& #2 & \makecell[t]{#3}  	\\}
{65}{$A^{65}$:			& $a$	& $a$	& $c$	& $abc$	& $d$	& $bcd$	& #2 & \makecell[t]{#3}  	\\}
{66}{$A^{66}$:			& $a$	& $a$	& $c$	& $d$	& $d$	& $ad$	& #2 & \makecell[t]{#3}  	\\}
{67}{$A^{67}$:			& $a$	& $a$	& $c$	& $d$	& $d$	& $abd$	& #2 & \makecell[t]{#3}  	\\}
{68}{$A^{68}$:			& $a$	& $a$	& $c$	& $d$	& $d$	& $cd$	& #2 & \makecell[t]{#3}  	\\}
{69}{$A^{69}$:			& $a$	& $a$	& $c$	& $d$	& $d$	& $bcd$	& #2 & \makecell[t]{#3}  	\\}
{70}{$A^{70}$:			& $a$	& $a$	& $c$	& $d$	& $ad$	& $cd$	& #2 & \makecell[t]{#3}  	\\}
{71}{$A^{71}$:			& $a$	& $a$	& $c$	& $d$	& $abd$	& $bcd$	& #2 & \makecell[t]{#3}  	\\}
{72}{$A^{72}$:			& $a$	& $a$	& $ac$	& $abc$	& $d$	& $bd$	& #2 & \makecell[t]{#3}  	\\}
{73}{$A^{73}$:			& $a$	& $a$	& $ac$	& $d$	& $bd$	& $bd$	& #2 & \makecell[t]{#3}  	\\}
{74}{$A^{74}$:			& $a$	& $a$	& $abc$	& $d$	& $d$	& $cd$	& #2 & \makecell[t]{#3}  	\\}
{75}{$A^{75}$:			& $a$	& $b$	& $b$	& $d$	& $d$	& $ad$	& #2 & \makecell[t]{#3}  	\\}
{76}{$A^{76}$:			& $a$	& $b$	& $b$	& $d$	& $ad$	& $cd$	& #2 & \makecell[t]{#3}  	\\}
{77}{$A^{77}$:			& $a$	& $b$	& $ab$	& $d$	& $d$	& $d$	& #2 & \makecell[t]{#3}  	\\}
{78}{$A^{78}$:			& $a$	& $b$	& $ab$	& $d$	& $d$	& $ad$	& #2 & \makecell[t]{#3}  	\\}
{79}{$A^{79}$:			& $a$	& $b$	& $ab$	& $d$	& $d$	& $cd$	& #2 & \makecell[t]{#3}  	\\}
{80}{$A^{80}$:			& $a$	& $b$	& $d$	& $d$	& $d$	& $ad$	& #2 & \makecell[t]{#3}  	\\}
{81}{$A^{81}$:			& $a$	& $b$	& $d$	& $d$	& $d$	& $bd$	& #2 & \makecell[t]{#3}  	\\}
{82}{$A^{82}$:			& $a$	& $b$	& $d$	& $d$	& $ad$	& $bd$	& #2 & \makecell[t]{#3}  	\\}
{83}{$A^{83}$:			& $a$	& $b$	& $d$	& $d$	& $ad$	& $cd$	& #2 & \makecell[t]{#3}  	\\}
{84}{$A^{84}$:			& $a$	& $b$	& $d$	& $d$	& $cd$	& $bcd$	& #2 & \makecell[t]{#3}  	\\}
{85}{$A^{85}$:	    	& $a$	& $b$	& $d$	& $ad$	& $cd$	& $bcd$	& #2 & \makecell[t]{#3}  	\\}
{86}{$A^{86}$:			& $a$	& $ab$	& $d$	& $d$	& $d$	& $bd$	& #2 & \makecell[t]{#3}  	\\}
{87}{$A^{87}$:			& $a$	& $ab$	& $d$	& $d$	& $cd$	& $bcd$	& #2 & \makecell[t]{#3}  	\\}
{88}{$A^{88}$:			& $a$	& $c$	& $c$	& $d$	& $d$	& $ad$	& #2 & \makecell[t]{#3}  	\\}
{89}{$A^{89}$:			& $a$	& $c$	& $c$	& $d$	& $d$	& $abd$	& #2 & \makecell[t]{#3}  	\\}
{90}{$A^{90}$:			& $a$	& $c$	& $ac$	& $d$	& $d$	& $d$	& #2 & \makecell[t]{#3}  	\\}
{91}{$A^{91}$:			& $a$	& $c$	& $ac$	& $d$	& $d$	& $ad$	& #2 & \makecell[t]{#3}  	\\}
{92}{$A^{92}$:			& $a$	& $c$	& $ac$	& $d$	& $d$	& $abd$	& #2 & \makecell[t]{#3}  	\\}
{93}{$A^{93}$:			& $a$	& $c$	& $abc$	& $d$	& $d$	& $ad$	& #2 & \makecell[t]{#3}  	\\}
{94}{$A^{94}$:			& $a$	& $c$	& $abc$	& $d$	& $d$	& $abd$	& #2 & \makecell[t]{#3}  	\\}
{95}{$A^{95}$:			& $a$	& $c$	& $d$	& $d$	& $d$	& $ad$	& #2 & \makecell[t]{#3}  	\\}
{96}{$A^{96}$:			& $a$	& $c$	& $d$	& $d$	& $d$	& $abd$	& #2 & \makecell[t]{#3}  	\\}
{97}{$A^{97}$:			& $a$	& $c$	& $d$	& $d$	& $d$	& $cd$	& #2 & \makecell[t]{#3}  	\\}
{98}{$A^{98}$:			& $a$	& $c$	& $d$	& $d$	& $d$	& $bcd$	& #2 & \makecell[t]{#3}  	\\}
{99}{$A^{99}$:			& $a$	& $c$	& $d$	& $d$	& $ad$	& $cd$	& #2 & \makecell[t]{#3}  	\\}
{100}{$A^{100}$:		& $a$	& $c$	& $d$	& $d$	& $abd$	& $bcd$	& #2 & \makecell[t]{#3}  	\\}
{101}{$A^{101}$:		& $a$	& $ac$	& $d$	& $d$	& $d$	& $cd$	& #2 & \makecell[t]{#3}  	\\}
{102}{$A^{102}$:		& $a$	& $ac$	& $d$	& $d$	& $d$	& $bcd$	& #2 & \makecell[t]{#3}  	\\}
{103}{$A^{103}$:		& $b$	& $ab$	& $d$	& $d$	& $d$	& $ad$	& #2 & \makecell[t]{#3}  	\\}
{104}{$A^{104}$:		& $b$	& $ab$	& $d$	& $d$	& $ad$	& $cd$	& #2 & \makecell[t]{#3}  	\\}
{105}{$A^{105}$:		& $c$	& $ac$	& $d$	& $d$	& $d$	& $ad$	& #2 & \makecell[t]{#3}  	\\}
{106}{$A^{106}$:		& $c$	& $ac$	& $d$	& $d$	& $d$	& $abd$	& #2 & \makecell[t]{#3}  	\\}
    }
}%

\newcounter{ct}
\begin{table}[t]
\small
\onecolumn
    \centering
\begin{tabular}{ccccccccc}
\forloop{ct}{1}{\value{ct} < 54}%
{%
	\R{\arabic{ct}}{}{}
}
\end{tabular}
\begin{tabular}{ccccccccc}
\forloop{ct}{54}{\value{ct} < 107}%
{%
	\R{\arabic{ct}}{}{}
}
\end{tabular}
    \caption{Profiles used for the proof of \Cref{cor:induction_basis_weak}.}
    \label{tab:profiles}
\end{table}

\onecolumn
\noindent\textbf{Case 1: $f(A^2)=[a,c,c]$}

As first case, we suppose that $f(A^1)=[a,c,c]$ and use a case distinction with respect to $A^{18}$ to derive an impossibility. Our first derivation hence shows that there are only two possible outcomes for this profile: $[c,c,c]$ and $[a,c,c]$. By showing that both cases are not possible, we disprove that $f(A^2)=[a,c,c]$.
\begin{longtable}{ccccccccc}
& $V_1$ & $V_2$ & $V_3$ & $V_4$ & $V_5$ & $V_6$ & Possible outcomes & Reason\\\toprule
\R{2}{$[a,c,c]$}{Assumption (Case 1)}\hline
\R{8}{$[a,c,c]$}{\PO{8}{b} \aux{8}{ab}{1} \SPg{8}{1}{c}{2}}\hline
\R{18}{$[a,c,c]$, $[c,c,c]$}{\PO{18}{b} \aux{18}{cd}{2} \SPg{18}{8}{abc}{3}}
\end{longtable}

\textbf{Case 1.1.: $f(A^{18})\neq [c,c,c]$}
\begin{longtable}{ccccccccc}
& $V_1$ & $V_2$ & $V_3$ & $V_4$ & $V_5$ & $V_6$ & Possible outcomes & Reason\\\toprule
\R{2}{$[a,c,c]$}{Assumption (Case 1)}\hline
\R{18}{$[c,c,c]$}{Assumption (for contradiction)}\hline
\R{10}{$[c,c,c]$}{\SPg{10}{18}{c}{3}}\hline
\R{29}{$[c,c,d]$, $[c,d,d]$}{\PO{29}{b} \WR{29}{c} \aux{29}{ad}{1} \SPl{29}{10}{a}{0}}\hline
\R{12}{$[c,c,c]$}{\SPg{12}{18}{c}{3}}\hline
\R{14}{$[a,c,c]$, $[a,a,c]$}{\PO{14}{b} \WR{14}{c} \aux{14}{ad}{1} \SPl{14}{12}{d}{0}}\hline
\R{11}{$[a,c,d]$}{\PO{11}{b} \WR{11}{c} $x_{abc}^{11}\leq x_{abc}^{29}\leq 2$ (SP from $A^{29}$ to $A^{11}$)\\ $x_{cd}^{11}\leq x_{cd}^{14}\leq 2$ (SP from $A^{14}$ to $A^{11}$)}\hline
\R{6}{$[a,c,c]$}{\PO{6}{b} \aux{6}{ab}{1} \SPg{6}{2}{c}{2}}\hline
\R{17}{$[a,c,d]$}{\PO{17}{b} \SPl{17}{11}{c}{1} \SPl{17}{6}{ab}{1} \SPg{17}{11}{abc}{2}}\hline
\R{7}{$[a,c,d]$}{\PO{7}{b} \WR{7}{c} \aux{7}{ab}{1} \SPl{7}{17}{abc}{2}}\hline
\R{8}{$[a,c,c]$}{\PO{8}{b} \aux{8}{ab}{1} \SPg{8}{2}{c}{2}}\hline
\R{9}{$[a,c,c]$}{\PO{9}{b} \aux{9}{ab}{1} \SPl{9}{8}{d}{0} \SPg{9}{7}{cd}{2}}\hline
\R{5}{$[a,c,c]$}{\PO{5}{b} \aux{5}{ab}{1} \SPg{5}{9}{c}{2}}\hline
\R{13}{\Lightning}{\PO{13}{b} \PO{13}{d} \SPl{13}{5}{ab}{1} \SPg{13}{11}{ad}{2}}
\end{longtable}

\textbf{Case 1.2: $f(A^{18})\neq [a,c,c]$}
\begin{longtable}{ccccccccc}
& $V_1$ & $V_2$ & $V_3$ & $V_4$ & $V_5$ & $V_6$ & Possible outcomes & Reason\\\toprule
\R{1}{$[a,a,c]$, $[a,b,c]$}{Assumption (Symmetry breaking)}\hline
\R{2}{$[a,c,c]$}{Assumption (Case 1)}\hline
\R{18}{$[a,c,c]$}{Assumption (for contradiction)}\hline
\R{16}{$[a,c,d]$}{\PO{16}{b} \WR{16}{c} \WR{16}{d} \SPl{16}{18}{cd}{2}}\hline
\R{6}{$[a,c,c]$}{\PO{6}{b} \aux{6}{ab}{1} \SPg{6}{2}{c}{2}}\hline
\R{17}{$[a,c,d]$}{\PO{17}{b} \aux{17}{abc}{2} \SPg{17}{16}{ad}{2} \SPl{17}{6}{ab}{1}}\hline
\R{7}{$[a,c,d]$}{\PO{7}{b} \WR{7}{c} \aux{7}{ab}{1} \SPl{7}{17}{abc}{2}}\hline
\R{8}{$[a,c,c]$}{\PO{8}{b} \aux{8}{ab}{1} \SPg{8}{2}{c}{2}}\hline
\R{9}{$[a,c,c]$}{\PO{9}{b} \aux{9}{ab}{1} \SPl{9}{8}{d}{0} \SPg{9}{7}{cd}{2}}\hline
\R{5}{$[a,c,c]$}{\PO{5}{b} \aux{5}{ab}{1} \SPg{5}{9}{c}{2}}\hline
\R{13}{$[a,c,c]$}{\PO{13}{b} \PO{13}{d} \SPl{13}{5}{ab}{1} \SPg{13}{5}{ad}{1}}\hline
\R{11}{$[c,c,d]$}{\PO{11}{b} \SPl{11}{13}{ad}{1} \SPl{11}{17}{abc}{2}}\hline
\R{29}{$[c,c,d]$}{\PO{29}{b} \aux{29}{ad}{1} \SPl{29}{11}{a}{0} \SPg{29}{11}{abc}{2}}\hline
\R{15}{$[a,c,c]$}{\PO{15}{b} \SPg{15}{6}{bc}{2} \SPl{15}{18}{cd}{2}}\hline
\R{19}{$[a,c,c]$, $[c,c,d]$}{\PO{19}{b} \SPl{19}{15}{c}{2} \SPg{19}{15}{bc}{2}}\hline
\R{3}{\makecell[t]{$\mathcal{W}_3\setminus \{[b,b,b], [b,b,c],$ \\$[b,c,c],[c,c,c]\}$}}{\SPl{3}{19}{bc}{2}}\hline
\R{4}{$[a,c,c]$}{\SPl{4}{1}{d}{0} \SPl{4}{2}{ab}{1} $x_{bc}^4\leq x_{bc}^3\leq 2$ (SP from $A^3$ to $A^4$)}\hline
\R{26}{$[c,c,c]$}{\PO{26}{a} \SPg{26}{4}{ac}{3}}\hline
\R{24}{$[c,c,c]$}{\SPg{24}{26}{c}{3}}\hline
\R{20}{$[c,c,c]$}{\SPg{20}{24}{c}{3}}\hline
\R{30}{$[c,c,d]$}{\PO{30}{a} \aux{30}{bd}{1} \SPl{30}{20}{b}{0} \SPl{30}{29}{ad}{1}}\hline
\R{25}{$[c,c,c]$}{\SPg{25}{26}{c}{3}}\hline
\R{21}{$[c,c,c]$}{\SPg{21}{25}{c}{3}}\hline
\R{23}{$[b,b,c]$, $[b,c,c]$}{\PO{23}{a} \WR{23}{c} \aux{23}{bd}{1} \SPl{23}{21}{d}{0}}\hline
\R{22}{\makecell[t]{$[b,b,c]$, $[b,c,c]$,\\ $[b,c,d]$}}{\PO{22}{a} \WR{22}{c} $x_{cd}^{22}\leq x_{cd}^{23}\leq 2$ (SP from $A^{23}$ to $A^{22}$)}\hline
\R{28}{$[a,c,c]$}{\PO{28}{b} \PO{28}{d} \aux{28}{abd}{1} \SPl{28}{5}{a}{1}}\hline
\R{27}{$[a,c,c]$, $[b,c,c]$}{\SPl{27}{28}{abd}{1} $x_{ab}^{27}\geq x_{ab}^{22}\geq 1$ (SP from $A^{27}$ to $A^{22}$)}\hline
\R{31}{\Lightning}{\SPg{31}{29}{bd}{1} \SPg{31}{30}{ad}{1} \SPg{31}{27}{abc}{3} \SPl{31}{27}{ab}{1}}
\end{longtable}

\noindent\textbf{Case 2: $f(A^2)=[a,a,c]$}

For the second case, we assume that $f(A^2)=[a,a,c]$. Subsequently, we infer the committees for multiple auxiliary profiles before we can derive a contradiction.\medskip

\textbf{Step 2.1: $f(A^{62})=[a,a,d]$}
\begin{longtable}{ccccccccc}
& $V_1$ & $V_2$ & $V_3$ & $V_4$ & $V_5$ & $V_6$ & Possible outcomes & Reason\\\toprule
\R{2}{$[a,a,c]$}{Assumption (Case 2)}\hline
\R{64}{$[a,a,c]$, $[a,a,d]$}{\aux{64}{cd}{1} \SPg{64}{2}{a}{2}}\hline
\R{61}{$[a,a,d]$}{\PO{61}{b} \WR{61}{d} \SPl{61}{64}{cd}{1}}\hline
\R{62}{$[a,a,a]$}{\aux{62}{abc}{2} \SPg{62}{61}{ad}{3} $f(A^{62})\neq[a,a,d]$ (Contradiction assumption)}\hline
\R{32}{$[a,a,a]$}{\SPg{32}{62}{a}{3}}\hline
\R{36}{$[a,a,d]$}{\aux{36}{cd}{1} \SPl{36}{32}{c}{0} \SPg{36}{64}{a}{2}}\hline
\R{34}{$[a,a,d]$}{\SPl{34}{36}{abc}{2} \SPg{34}{64}{a}{2}}\hline
\R{33}{$[a,a,a]$}{\SPg{33}{62}{a}{3}}\hline
\R{35}{\Lightning}{\aux{35}{cd}{1} \SPl{35}{33}{d}{0} \SPg{35}{34}{ad}{3}}
\end{longtable}

\textbf{Step 2.2: $f(A^{68})=[a,c,d]$}
\begin{longtable}{ccccccccc}
& $V_1$ & $V_2$ & $V_3$ & $V_4$ & $V_5$ & $V_6$ & Possible outcomes & Reason\\\toprule
\R{68}{$[a,d,d]$}{\WR{68}{a} \WR{68}{d} \aux{68}{cd}{2} $f(A^{68})\neq [a,c,d]$ (Contradiciton assumption)}\hline
\R{2}{$[a,a,c]$}{Assumption (Case 2)}\hline
\R{62}{$[a,a,d]$}{Step 2.1}\hline
\R{59}{$[a,c,d]$}{\WR{59}{a} \WR{59}{c} \SPl{59}{62}{abc}{2}}\hline
\R{70}{$[a,d,d]$}{\WR{70}{a} \SPg{70}{59}{cd}{2} \SPg{70}{68}{ad}{3}}\hline
\R{66}{$[a,d,d]$}{\WR{66}{a} \SPg{66}{70}{d}{2}}\hline
\R{64}{$[a,a,c]$, $[a,a,d]$}{\aux{64}{cd}{1} \SPg{64}{2}{a}{2}}\hline
\R{61}{$[a,a,d]$}{\PO{61}{b} \WR{61}{d} \SPl{61}{64}{cd}{1}}\hline
\R{93}{$[a,d,d]$}{\aux{93}{abc}{1} \SPl{93}{66}{a}{1} \SPg{93}{61}{ad}{3}}\hline
\R{88}{$[c,d,d]$}{\WR{88}{c} \SPl{88}{93}{abc}{1}}\hline
\R{99}{$[d,d,d]$}{\SPg{99}{88}{cd}{3} \SPg{99}{68}{ad}{3}}\hline
\R{95}{$[d,d,d]$}{\SPg{95}{99}{d}{3}}\hline
\R{97}{$[d,d,d]$}{\SPg{97}{99}{d}{3}}\hline
\R{101}{$[a,d,d]$, $[a,a,d]$}{\PO{101}{b} \WR{101}{d} \aux{101}{ac}{1} \SPl{101}{97}{c}{0}}\hline
\R{91}{$[a,d,d]$, $[c,d,d]$}{\aux{91}{ac}{1} \SPl{91}{93}{abc}{1}}\hline
\R{90}{$[a,d,d]$}{\aux{90}{ac}{1} $x_{cd}^{90}\leq x_{cd}^{101}\leq 2$ (SP from $A^{101} to A^{90}$)\\ \SPg{90}{91}{d}{2}}\hline
\R{105}{\Lightning}{\aux{105}{ac}{1} \SPl{105}{95}{a}{0} \SPg{105}{90}{ad}{3}}
\end{longtable}

\textbf{Step 2.3: Deriving auxiliary profiles}

As next step, we infer the outcomes for multiple profiles based on the previous insights. Since we use no contradiction assumption in this step, the derived outcomes can be used in the subsequent deductions. 
\begin{longtable}{ccccccccc}
& $V_1$ & $V_2$ & $V_3$ & $V_4$ & $V_5$ & $V_6$ & Possible outcomes & Reason\\\toprule
\R{2}{$[a,a,c]$}{Assumption (Case 2)}\hline
\R{62}{$[a,a,d]$}{Step 2.1}\hline
\R{68}{$[a,c,d]$}{Step 2.2}\hline
\R{59}{$[a,c,d]$}{\WR{59}{a} \WR{59}{c} \SPl{59}{62}{abc}{2}}\hline
\R{70}{$[a,c,d]$}{\WR{70}{a} \SPl{70}{62}{abc}{2} \SPl{70}{68}{d}{1} \SPg{70}{59}{cd}{2}}\hline
\R{64}{$[a,a,c]$}{\PO{64}{b} \SPl{64}{70}{ad}{2} \SPg{64}{2}{a}{2}}\hline
\R{61}{$[a,a,d]$}{\PO{61}{b} \WR{61}{d} \SPl{61}{64}{cd}{1}}\hline
\R{74}{$[a,a,d]$}{\PO{74}{b} \WR{74}{d} \SPl{74}{61}{c}{0} \SPg{74}{68}{abc}{2}}\hline
\R{57}{$[a,a,d]$}{\PO{57}{b} \PO{57}{c} \WR{57}{d} \SPg{57}{74}{ab}{2}}\hline
\R{38}{$[a,a,d]$}{\PO{38}{b} \PO{38}{c} \WR{38}{d} \SPg{38}{74}{a}{2}}\hline
\R{37}{$[a,a,d]$}{\PO{37}{b} \PO{37}{c} \SPl{37}{38}{cd}{1} \SPg{37}{38}{bd}{1}}\hline
\R{39}{$[a,a,d]$}{\PO{39}{b} \PO{39}{c} \SPl{39}{37}{d}{1} \SPg{39}{37}{bd}{1}}\hline
\R{73}{$[a,a,d]$}{\PO{73}{b} \PO{73}{c} \SPl{73}{39}{a}{2} \SPg{73}{39}{ac}{2}}\hline
\R{46}{$[a,a,b]$, $[a,a,d]$}{\PO{46}{c} \aux{46}{bd}{1} \SPl{73}{46}{bd}{1}}\hline
\R{34}{$[a,a,c]$}{\PO{34}{b} \aux{34}{cd}{1} \SPl{34}{70}{ad}{2} \SPg{34}{64}{a}{2}}\hline
\R{36}{$[a,a,c]$}{\PO{36}{b} \aux{36}{cd}{1} \SPl{36}{34}{c}{1} \SPg{36}{34}{abc}{3}}\hline
\R{32}{$[a,a,c]$}{\PO{32}{b} \SPg{32}{36}{c}{1} \SPg{32}{64}{a}{2}}\hline
\R{55}{$[a,a,c]$, $[a,a,d]$}{\PO{55}{b} \SPl{55}{32}{a}{2} \SPg{55}{32}{ab}{2}}\hline
\R{63}{$[a,a,d]$}{\PO{63}{b} \SPl{63}{32}{a}{2} \SPl{63}{61}{d}{1} \SPg{63}{61}{abd}{3}}\hline
\R{60}{$[a,c,d]$}{\PO{60}{b} \WR{60}{a} \WR{60}{c} \SPl{60}{63}{abc}{2}}\hline
\R{71}{$[a,c,d]$, $[a,d,d]$}{\PO{71}{b} \WR{71}{a} \SPl{71}{63}{abc}{2} \SPg{71}{60}{bcd}{2}}\hline
\R{65}{$[a,a,c]$, $[a,a,d]$}{\PO{65}{b} \aux{65}{bcd}{1} \SPl{65}{64}{cd}{1}}\hline
\R{42}{$[a,a,c]$, $[a,a,d]$}{\SPl{42}{65}{bcd}{1} \SPl{42}{55}{ab}{2}}
\end{longtable}

\textbf{Step 2.4: $f(A^{48})=[a,a,d]$}

Next, we prove that $f(A^{48})=[a,a,d]$. For this, we show first by contradiction that $f(A^{48})\in\{[a,a,a], [a,a,d]\}$.
\begin{longtable}{ccccccccc}
& $V_1$ & $V_2$ & $V_3$ & $V_4$ & $V_5$ & $V_6$ & Possible outcomes & Reason\\\toprule
\R{48}{$[a,a,b]$, $[a,b,d]$}{\PO{48}{c} \WR{48}{a} \aux{48}{ad}{2} \aux{48}{abc}{2} $f(A^{48},3)\neq [a,a,a]$ (Contradiction assumption)\\$f(A^{48},3)\neq [a,a,d]$ (Contradiction assumption)}\hline
\R{44}{\makecell[t]{$[a,a,b]$, $[a,b,b]$,\\ $[a,b,d]$}}{\PO{44}{c} \WR{44}{a} \SPl{44}{48}{ad}{2}}\hline
\R{42}{\Lightning}{$x_{ac}^{42}\leq x_{ac}^{44}\leq 2$ (SP from $A^{44}$ to $A^{42}$)\\ \SPl{48}{44}{ad}{2} $f(A^{42})\in \{[a,a,c], [a,a,d]\}$ (Step 2.3)}
\end{longtable}

As second step, we show by contradiction that $f(R^{48})\neq [a,a,a]$.
\begin{longtable}{ccccccccc}
& $V_1$ & $V_2$ & $V_3$ & $V_4$ & $V_5$ & $V_6$ & Possible outcomes & Reason\\\toprule
\R{48}{$[a,a,a]$}{Assumption (for contradiction)}\hline
\R{44}{$[a,a,a]$}{\PO{44}{c} \SPg{44}{48}{ac}{3}}\hline
\R{72}{$[a,a,d]$}{\PO{72}{c} \aux{72}{abc}{2} \aux{72}{bd}{1} \SPl{72}{44}{b}{0}}\hline
\R{46}{$[a,a,d]$}{\SPl{46}{72}{abc}{2} $f(A^{46},3)\in \{[a,a,b], [a,a,d]\}$ (Step 2.3)}\hline
\R{45}{$[a,a,a]$}{\PO{45}{c} \SPg{45}{48}{ac}{3}}\hline
\R{47}{\Lightning}{\aux{47}{bd}{1} \SPl{45}{47}{d}{0} \SPg{45}{46}{ad}{3}}
\end{longtable}

\textbf{Step 2.5: $f(A^{52})=[a,b,d]$}
\begin{longtable}{ccccccccc}
& $V_1$ & $V_2$ & $V_3$ & $V_4$ & $V_5$ & $V_6$ & Possible outcomes & Reason\\\toprule
\R{48}{$[a,a,d]$}{Step 2.4}\hline
\R{40}{$[a,b,d]$}{\WR{40}{a} \WR{40}{b} \SPl{40}{48}{abc}{2}}\hline
\R{52}{$[a,d,d]$}{\WR{52}{a} \SPl{52}{40}{b}{1} \SPg{52}{40}{bd}{2} $f(A^{52})\neq[a,b,d]$ (Contradiction assumption)}\hline
\R{50}{$[a,d,d]$}{\WR{50}{a} \SPg{50}{52}{d}{2}}\hline
\R{57}{$[a,a,d]$}{Step 2.3}\hline
\R{56}{$[a,a,d]$}{\PO{56}{c} \WR{56}{d} \SPl{56}{57}{cd}{1} \SPl{56}{50}{b}{0}}\hline
\R{41}{$[a,a,d]$}{\PO{41}{c} \WR{41}{d} \SPl{41}{56}{bd}{1}}\hline
\R{49}{$[a,d,d]$}{\WR{49}{a} \SPg{49}{52}{d}{2}}\hline
\R{78}{$[a,d,d]$}{\aux{78}{ab}{1} \SPl{78}{49}{a}{1} \SPg{78}{41}{ad}{3}}\hline
\R{75}{$[b,d,d]$}{\PO{75}{c} \WR{75}{b} \SPl{75}{78}{ab}{1}}\hline
\R{82}{$[d,d,d]$}{\SPg{82}{75}{bd}{3} \SPg{82}{50}{ad}{3}}\hline
\R{81}{$[d,d,d]$}{\SPg{81}{82}{d}{3}}\hline
\R{86}{$[a,d,d]$, $[a,a,d]$}{\PO{86}{c} \WR{86}{d} \aux{86}{ab}{1} \SPl{86}{81}{b}{0}}\hline
\R{77}{$[a,d,d]$}{\PO{77}{c} $x_{bd}^{77}\leq x_{bd}^{86}\leq 2$ (SP from $A^{86}$ to $A^{77}$)\\ \SPg{77}{78}{d}{2}}\hline
\R{80}{$[d,d,d]$}{\SPg{80}{82}{d}{3}}\hline
\R{103}{\Lightning}{\aux{103}{ab}{1} \SPl{103}{80}{a}{0} \SPg{103}{77}{ad}{3}}
\end{longtable}

\textbf{Step 2.6: $f(A^{71})=[a,d,d]$}
\begin{longtable}{ccccccccc}
& $V_1$ & $V_2$ & $V_3$ & $V_4$ & $V_5$ & $V_6$ & Possible outcomes & Reason\\\toprule
\R{71}{$[a,c,d]$}{$f(A^{71})\in\{[a,c,d], [a,d,d]\}$ (Step 2.3)\\ $f(A^{71})\neq [a,d,d]$ (Contradiction assumption)}\hline
\R{52}{$[a,b,d]$}{Step 2.5}\hline
\R{46}{$[a,a,b]$}{\SPl{46}{52}{ad}{2} $f(A^{46})\in\{[a,a,b], [a,a,d]\}$ (Step 2.3)}\hline
\R{72}{$[a,a,b]$}{\PO{72}{c} \aux{72}{bd}{1} \SPl{72}{46}{b}{1} \SPg{72}{46}{abc}{3}}\hline
\R{44}{$[a,a,b]$}{\PO{44}{c} \SPl{44}{72}{bd}{1} \SPg{44}{72}{b}{1}}\hline
\R{42}{$[a,a,d]$}{\SPl{42}{44}{ac}{2} $f(A^{42})\in \{[a,a,c], [a,a,d]\}$ (Step 2.3)}\hline
\R{43}{$[a,c,d]$}{\WR{71}{a} \SPl{43}{42}{abc}{2} \SPl{43}{71}{abd}{2}}\hline
\R{54}{$[a,d,d]$}{\PO{54}{c} \WR{54}{a} \SPg{54}{43}{cd}{2}}\hline
\R{48}{$[a,a,d]$}{Step 2.4}\hline
\R{40}{$[a,b,d]$}{\PO{40}{c} \WR{40}{a} \WR{40}{b} \SPl{40}{48}{abc}{2}}\hline
\R{53}{$[a,d,d]$}{\WR{53}{a} \SPg{53}{40}{bcd}{2} \SPg{53}{54}{ad}{3}}\hline
\R{49}{$[a,d,d]$}{\WR{49}{a} \SPg{49}{53}{d}{2}}\hline
\R{51}{$[a,d,d]$}{\WR{51}{a} \SPg{51}{53}{d}{2}}\hline
\R{57}{$[a,a,d]$}{Step 2.3}\hline
\R{58}{$[a,a,d]$}{\PO{58}{c} \WR{58}{d} \SPl{58}{51}{b}{0} \SPl{58}{57}{cd}{1}}\hline
\R{41}{$[a,a,d]$}{\WR{41}{d} \SPl{41}{58}{bcd}{1}}\hline
\R{78}{$[a,d,d]$}{\aux{78}{ab}{1} \SPl{78}{49}{a}{1} \SPg{78}{41}{ad}{3}}\hline
\R{75}{$[b,d,d]$}{\PO{75}{c} \WR{75}{b} \SPl{75}{78}{ab}{1}}\hline
\R{76}{$[b,d,d]$}{\PO{76}{c} \WR{76}{b} \SPg{76}{75}{cd}{2}}\hline
\R{85}{$[d,d,d]$}{\SPg{85}{76}{bcd}{3} \SPg{85}{54}{ad}{3}}\hline
\R{83}{$[d,d,d]$}{\SPg{83}{85}{d}{3}}\hline
\R{104}{$[b,d,d]$}{\PO{104}{c} \aux{104}{ab}{1} \SPl{104}{83}{a}{0} \SPg{104}{78}{cd}{2}}\hline
\R{79}{$[b,d,d]$}{\PO{79}{c} \SPl{79}{104}{ad}{2} \SPg{79}{78}{cd}{2}}\hline
\R{87}{$[b,d,d]$}{\PO{87}{c} \aux{87}{ab}{1} \SPl{87}{79}{b}{1} \SPg{87}{79}{bcd}{3}}\hline
\R{84}{$[b,d,d]$}{\PO{84}{c} \SPl{84}{87}{ab}{1} \SPg{84}{87}{b}{1}}\hline
\R{85}{\Lightning}{\SPl{85}{84}{d}{2} \SPg{85}{76}{bcd}{3} \SPg{85}{54}{ad}{3}}
\end{longtable}

\textbf{Step 2.7: The final contradiction}
\begin{longtable}{ccccccccc}
& $V_1$ & $V_2$ & $V_3$ & $V_4$ & $V_5$ & $V_6$ & Possible outcomes & Reason\\\toprule
\R{71}{$[a,d,d]$}{Step 2.6}\hline
\R{67}{$[a,d,d]$}{\WR{67}{a} \SPg{67}{71}{d}{2}}\hline
\R{61}{$[a,a,d]$}{Step 2.3}\hline
\R{94}{$[a,d,d]$}{\PO{94}{b} \aux{94}{abc}{1} \SPl{94}{67}{a}{1} \SPg{94}{61}{abd}{3}}\hline
\R{89}{$[c,d,d]$}{\WR{89}{c} \SPl{89}{94}{abc}{1}}\hline
\R{69}{$[a,d,d]$}{\WR{69}{a} \SPg{69}{71}{d}{2}}\hline
\R{100}{$[d,d,d]$}{\PO{100}{b} \SPg{100}{69}{abd}{3} \SPg{100}{89}{bcd}{3}}\hline
\R{98}{$[d,d,d]$}{\SPg{98}{100}{d}{3}}\hline
\R{102}{$[a,a,d]$, $[a,d,d]$}{\PO{102}{b} \WR{102}{d} \aux{102}{ac}{1} \SPl{102}{98}{c}{0}}\hline
\R{92}{$[a,d,d]$, $[c,d,d]$}{\aux{92}{ac}{1} \SPl{94}{92}{abc}{1}}\hline
\R{90}{$[a,d,d]$}{$x_{bcd}^{90}\leq x_{bcd}^{102}\leq 2$ (SP from $A^{102}$ to $A^{90}$)\\
\SPg{90}{92}{d}{2}}\hline
\R{96}{$[d,d,d]$}{\SPg{96}{100}{d}{3}}\hline
\R{106}{\Lightning}{\aux{106}{ac}{1} \SPl{106}{90}{c}{0} \SPl{106}{96}{a}{0}}
\end{longtable}
\end{proof}

\twocolumn
\subsection{Proof of \Cref{thm:wSP}}

Finally, we prove \Cref{thm:wSP}.

\wSP*
\begin{proof}
Just as in the main body for Thiele rules, we present here constructions for sequential Thiele rules and divisor methods based on majoritarian portioning, proving that these rules fail strategy\-proofness for unrepresented voters.\medskip

\textbf{Sequential Thiele rules.}
Consider any sequential $w$-Thiele rule $f$ other than approval voting. Since, the vector $w$ is decreasing and $f$ is not approval voting, there is an index $j$ such that $w_j<1$. Let $j^*$ denote the first such index. Moreover we define $\ell \in \mathbb{N}$, $\ell \geq 4$, as the smallest integer such that $w_{j^*} < \frac{\ell-2}{\ell}$. Finally, consider the following two profiles $A$ and $A'$ with $m=4$ parties and $n = 4\cdot \ell + 1$ voters (the numbers before the preference relation indicate how often a preference relation is reported, e.g., $\ell$ voters approve the set $ab$ in $A$), 
\begin{center}
    \begin{tabular}{cccccccc}
        $A$: & $1$: $b$ & $\ell$: $ab$ & $\ell$: $bd$ & $\ell$: $ac$ & $\ell-1$: $cd$ & $1$: $d$  \\
        $A'$: & $1$: $b$ & $\ell$: $ab$ & $\ell$: $bd$ & $\ell$: $ac$ & $\ell-1$: $cd$ & $1$: $ad$
    \end{tabular}
\end{center}
We will now show that an unrepresented voter can manipulate $f$ in $A$ if $k=j^*$. Note for this that $w_1=\dots=w_{j^*-1}=1$, which implies that $f$ assigns the first $j^*-1$ seats to the approval winner. For the profile $A$, this means that these seats go to party $b$ as it is approved by $2\ell+1$ voters. Finally, the last seat goes to party $c$. For proving this claim, let $W^x$ denote the committee which assigns $j^*-1$ seats to party $b$ and the last seat to party $x$. Then, $s(W^c,A)=(j^*-1)(2\ell+1)+2\ell-1$, $s(W^a,A)=s(W^d,A)=(j^*-1)(2\ell+1)+\ell \cdot (1 + w_{j^*}) <(j^*-1)(2\ell+1)+ \ell (1+\frac{\ell-2}{\ell}) = (j^*-1)(2\ell+1)+2\ell -2$, and $s(W^b, A)=(j^*-1)(2\ell+1)+(2\ell+1)w_{j^*}<(j^*-1)(2\ell+1)+(2\ell+1)\frac{\ell-2}{\ell}<(j^*-1)(2\ell+1)+2\ell-3$, which proves our claim. 

Next, consider the profile $A'$. Just as for $A$, $f$ assigns the first $j^*-1$ seats to the approval winner, which is in this case $a$ as both $a$ and $b$ are approved by $2\ell+1$ voters and the lexicographic tie-breaking chooses $a$. Finally, analogous computations as for $A$ show that the last seat then goes to $d$, i.e., $f(A',j^*)$ chooses the committee that assigns $j^*-1$ seats to $a$ and one seat to $d$. Since $A$ and $A'$ differ only in the preference of the last voter (who approves only $d$) and $f(A,j^*,d)=0 < 1=f(A',j^*,d)$, this proves that an unrepresented voter can manipulate $f$. 
\medskip

\textbf{Divisor methods based on majoritarian portioning.}
Let $f$ denote a divisor method based on majoritarian portioning other than approval voting. This means that there is a monotone function $g:\mathbb{N}_0\rightarrow \mathbb{R}_{>0}$ such that $f$ can be computed as follows: in the $i$-th round, $f$ assigns the next seat to the party $x$ that maximizes $\frac{w_x}{g(t_x^i)}$ ($w_x$ denotes the weight of party $x$ computed by majoritarian portioning, and $t_x^i$ denotes the number of seats assigned to $x$ in all previous iterations). Now, since $f$ is not approval voting, there are integers $\ell, j\in\mathbb{N}$ such that $\frac{\ell+1}{g(j)}\leq \frac{\ell}{g(0)}$. Let $\ell^*$, $j^*$ denote a pair of such indices that minimize $j^*$, i.e., for all $j'<j^*$ and $\ell\in\mathbb{N}_0$, it holds that $\frac{\ell+1}{g(j')}> \frac{\ell}{g(0)}$.

As next step, we show that we may assume that $\ell^*\geq 2$. If this is not the case then, $\ell^*=1$. Note that $\ell^*=0$ is impossible as $\frac{1}{g(j)}> \frac{0}{g(0)}$. Now, if $\ell^*=1$, our conditions require that $\frac{2}{g(j^*)}\leq \frac{1}{g(0)}$ and $\frac{2}{g(j)}> \frac{1}{g(0)}$ for every $j\in\mathbb{N}_0$ with $j<j^*$. We will show subsequently that we can set $\ell^*$ to $3$. In particular, we have  $\frac{4}{g(j^*)}\leq \frac{2}{g(0)}<\frac{3}{g(0)}$, showing that our first condition is met. Next, consider the condition that $\frac{4}{g(j)}>\frac{3}{g(0)}$ for all $j<j^*$ and assume for contradiction that there is $j'\in\mathbb{N}$ such that $j'<j^*$ and $\frac{4}{g(j')}\leq \frac{3}{g(0)}$. Since $\frac{4}{g(0)}> \frac{3}{g(0)}$, we can thus find an index $j''\in\mathbb{N}$ such that $j''<j^*$, $\frac{4}{g(j)}> \frac{3}{g(0)}$ for all $j<j''$ and $\frac{4}{g(j'')}\leq \frac{3}{g(0)}$. However, this contradicts the definition of $j^*$ as $j^*$ is defined to be minimal. Hence, we can set $\ell^*$ to $3$.

Finally, we define two profiles $A$ and $A'$ on $n=4\ell^* +4$ voters and $m=4$ parties on which $f$ fails strategy\-proofness for unrepresented voters if $k=j^*+1$. In the profiles, the numbers before the preference relation indicate how often a preference relation is reported, e.g., $\ell^*$ voters approve the set $ab$ in $A$. 
\begin{center}
\begin{tabular}{ccccccc}
    $A$:  & $2$: $c$ & $2$: $d$ & $\ell^*$: $ac$ & $\ell^*$: $cd$ & $\ell^*$: $ab$ & $\ell^*$: $bd$\\
    $A'$:  & $2$: $c$ & $2$: $ad$ & $\ell^*$: $ac$ & $\ell^*$: $cd$ & $\ell^*$: $ab$ & $\ell^*$: $bd$
\end{tabular}
\end{center}

In profile $A$ we see that in total, $a$ and $b$ are each approved by $2\ell^*$ voters and $c$ and $d$ are each approved by $2\ell^*+2$ voters, while in $A'$ two voters additionally approve $a$.
When applying majoritarian portioning (with lexicographic tie-breaking), we derive for $A$ that $w_a=0$, $w_b=2\ell^*$, $w_c=2\ell^*+2$, and $w_d=2$. In more detail, we first allocate $2\ell^*+2$ votes to $c$ because of the lexicographic tie-breaking. After removing all voters who approve $c$, the approval score of $d$ is $\ell^*+2$ and the approval score of $b$ is $2\ell^*$. Since $\ell^*\geq 2$, we next allocate $2\ell^*$ votes to $b$. Finally, only the two voters who approve $d$ remain, giving $d$ a weight of $2$. Now, using the definition of $\ell^*$ and $j^*$, it follows that $f$ assigns the first $j^*$ seats to $c$ because $\frac{\ell^*+1}{g(j')}>\frac{\ell^*}{g(0)}$ for all $j'<j^*$. In contrast, $\frac{\ell^*+1}{g(j^*)}\leq \frac{\ell^*}{g(0)}$ implies that the last seat goes to $b$. Hence, $f$ outputs a committee $W$ with $W(c)=j^*$ and $W(b)=1$

Analogous computations as for $A$ show that majoritarian portioning results in the following weights for $A'$: $w_a=2\ell^*+2$, $w_b=0$, $w_c=2$, $w_d=2\ell^*$. Then, a symmetric analysis as for $A$ shows that $f$ elects the committee $W'$ with $W'(a)=j^*$ and $W'(d)=1$. Finally, note that $A$ and $A'$ only differ in the ballot of the second type of voters. Since these voters approve only $d$ in $A$ but $W(d)=0$, strategy\-proofness for unrepresented voters requires that $d$ cannot obtain a seat if these voters manipulate. However, by letting them deviate one after another, the outcome changes eventually in their favor, which means that $f$ can be manipulated by unrepresented voters.
\end{proof}

\end{document}